\documentclass{llncs}
\usepackage[demo]{graphicx}
\usepackage[latin1]{inputenc}
\usepackage[T1]{fontenc}
\usepackage{xspace}
\usepackage[normalem]{ulem}
\usepackage[inline]{enumitem}
\setlistdepth{9}
\setlist[enumerate,1]{label=\roman*.}
\setlist[enumerate,2]{label=\alph*.}
\setlist[enumerate,3]{label=\arabic*.}
\setlist[enumerate,4]{label=(\roman*).}
\setlist[enumerate,5]{label=(\alph*).}
\setlist[enumerate,6]{label=(\arabic*).}
\setlist[enumerate,7]{label=(\Alph*).}
\setlist[enumerate,8]{label=[\roman*].}
\setlist[enumerate,9]{label=[\arabic*].}

\renewlist{enumerate}{enumerate}{9}

\usepackage{stmaryrd}
\usepackage[usenames,dvipsnames,table]{xcolor}
\usepackage{tikz}
\usetikzlibrary{spy,backgrounds,calc}
\usetikzlibrary{fixedpointarithmetic}
\usetikzlibrary{calc}
\usetikzlibrary{arrows}
\usetikzlibrary{decorations}
\usetikzlibrary{decorations.markings}
\usetikzlibrary{decorations.pathmorphing}
\usetikzlibrary{positioning}
\usetikzlibrary{shapes}
\usepackage{bm}
\usepackage{amssymb}
\usepackage{amsmath}

\usepackage{amsthm}
\usepackage{hyperref}
\usepackage{cleveref}
\usepackage{mathtools}
\usepackage{multirow}
\usepackage{minibox}

\newcommand{\maxdeg}{{\ensuremath{\Delta}}\xspace}
\newcommand{\mopt}{\ensuremath{{M^*}}\xspace}
\newcommand{\m}{\ensuremath{M}\xspace}
\newcommand{\oneOverTwo}{\ensuremath{{\frac{1}{2}}}\xspace}

\newcommand{\twoOverThree}{\ensuremath{{\frac{2}{3}}}\xspace}
\newcommand{\threeOverFour}{\ensuremath{{\frac{3}{4}}}\xspace}

\newcommand{\mingreedy}{\textsc{MinGreedy}\xspace}
\newcommand{\mds}{\textsc{MDS}\xspace}

\newcommand{\greedy}{\textsc{Greedy}\xspace}
\newcommand{\ranking}{\textsc{Ranking}\xspace}
\newcommand{\mrg}{\textsc{MRG}\xspace}
\newcommand{\karpsipser}{\textsc{KarpSipser}\xspace}
\newcommand{\shuffle}{\textsc{Shuffle}\xspace}

\newcommand{\g}{{\ensuremath{G}}\xspace}
\newcommand{\mg}{{\ensuremath{H}}\xspace}
\newcommand{\w}[1]{{\ensuremath{m_{#1}}}\xspace}
\newcommand{\wopt}[1]{{\ensuremath{m_{#1}^*}}\xspace}

\newcommand{\edge}[1]{\ensuremath{\{#1\}}\xspace}
\newcommand{\dedge}[1]{\ensuremath{(#1)}\xspace}

\newcommand{\inlineheading}[1]{\textbf{#1}}

\newcommand{\dataitem}[1]{\ensuremath{\langle #1\rangle}\xspace}

\newcommand{\transferred}{\ensuremath{\kappa}\xspace}
\newcommand{\hidethis}[1]{}

\newcommand{\theratio}{\ensuremath{\frac{\maxdeg}{2\maxdeg-2}}\xspace}
\newcommand{\invariant}{\textsc{Invariant}\xspace}

\crefname{definition}{Definition}{Definitions}
\crefname{lemma}{Lemma}{Lemmas}
\crefname{appendix}{Appendix}{Appendices}
\crefname{theorem}{Theorem}{Theorems}
\crefname{observation}{Observation}{Observations}
\crefname{equation}{\unskip}{\unskip}
\crefname{proposition}{Proposition}{Propositions}
\crefname{corollary}{Corollary}{Corollaries}
\crefname{table}{Table}{Tables}
\crefname{enumi}{\unskip}{\unskip}
\crefname{section}{Section}{Sections}
\crefname{figure}{Figure}{Figures}

\begin{document}

\mainmatter  

\title{On the Approximation Performance of Degree Heuristics for Matching}


%
%
\author{Bert Besser%
\thanks{Partially supported by DFG SCHN 503/6-1.}%
\and Bastian Werth}
%

\institute{Institut f\"ur Informatik, Goethe-Universit\"at Frankfurt am Main, Germany}

%
%

\maketitle


\begin{abstract}
In the design of greedy algorithms for the maximum cardinality matching problem the utilization of degree information when selecting the next edge is a well established and successful approach.

\smallskip

We define the class of ``degree sensitive'' greedy matching algorithms, which allows us to analyze many well-known heuristics, and provide tight approximation guarantees under worst case tie breaking.
We exhibit algorithms in this class with optimal approximation guarantee for bipartite graphs.
In particular the \karpsipser algorithm, which picks an edge incident with a degree-1 node if possible and otherwise an arbitrary edge, turns out to be optimal with approximation guarantee~\theratio, where~\maxdeg is the maximum degree.

\end{abstract}
\section{Introduction}
Matching problems occur in many applications such as online advertising \cite{msvv07}, image feature matching \cite{cwchr96}, or protein structure comparison \cite{bsx08}.

In Maximum Cardinality Matching a set of node-disjoint edges of maximum size is to be determined.
This problem can be solved in time~$ O(m\sqrt{n}) $ for bipartite as well as general graphs \cite{b90,gt91,hk73,v12}.
The~$ O(n^{2.5}) $ barrier was finally broken in \cite{ms04} with a runtime of~$ O(n^\omega) $, where~$ \omega<2.38 $ holds.

In scenarios where obtaining exact solutions is of less importance than ease of implementation and fast runtime, an approximate greedy algorithm is an adequate choice.
Moreover, greedy matchings can be used as input for exact algorithms to obtain considerable speed-ups \cite{lms10}.

The following randomized greedy algorithms can be implemented in linear time~$ O(n{+}m) $ \cite{frs93,ks81,m97,ps12}.
The \greedy algorithm \cite{t84} picks an edge which is node disjoint from all previously picked edges, the \karpsipser algorithm works like \greedy but picks an edge incident with a node of degree one, if such a node exists \cite{ks81}.
The \mrg algorithm (``modified random greedy'') \cite{t84} first selects a node and then matches it with a neighbor, its variation \mingreedy \cite{t84} first selects a node of minimum degree.
The \shuffle algorithm \cite{gt12} computes a permutation~$ \pi $, processes nodes according to~$ \pi $ and each time picks the $ \pi $-lexicographically first edge.
(\ranking \cite{kvv90} works similar to \shuffle but is tailored for an on-line setting in bipartite graphs.)

\medskip

\inlineheading{Previous Work.}
Experiments show that large matchings are produced by the above algorithms if ties are broken uniformly at random \cite{frs93,lms10,m97,t84}.

All mentioned algorithms compute \emph{maximal} matchings, i.e. matchings to which no further edge can be added.
A maximal matching is at least half as large as a \emph{maximum} matching, hence the above algorithms trivially achieve approximation ratio at least~\oneOverTwo. 
An expected approximation ratio larger than~\oneOverTwo, namely~$ \oneOverTwo{+}\frac{1}{400.000} $, was shown first for \mrg in \cite{adfs95}.
However, the best known inapproximability bound on the expected approximation ratio of \mrg is~$ \twoOverThree $, using methods in~\cite{df91}.
For \shuffle, only recently an expected approximation ratio of at least~$ \approx 0.523 $ was shown in~\cite{ccwz14}, whereas it is only known from \cite{gt12} that this ratio cannot be larger than~$ \frac{3}{4} $.
The expected approximation ratio of \greedy and \mingreedy is at most~$ \oneOverTwo+\varepsilon $, for any~$ \varepsilon>0 $ \cite{df91,p12}.

The expected performance on degree bounded graphs remains open for all mentioned algorithms.
On graphs with degrees at most three, no algorithm discussed so far achieves an expected approximation ratio better than~$ \frac{5}{6} $ \cite{p12}.
An expected approximation ratio of at least $ \oneOverTwo ( \sqrt{(\maxdeg - 1)^2 + 1} - \maxdeg + 2) $ is achieved by \greedy on graphs with degrees at most~\maxdeg \cite{m97}.

Furthermore \mingreedy leaves~$ o(n) $ nodes unmatched in large random 3-regular graphs \cite{frs93}.
In large sparse random graphs \karpsipser computes matchings within~$ o(n) $ of optimum size~\cite{afp98}.

Assuming worst case instead of random uniform tie breaking, in~\cite{b14} it is shown that \mingreedy is guaranteed to compute a matching of size at least~$ \frac{\maxdeg-1/2}{2\maxdeg-2} $ times optimal, if degrees are at most~\maxdeg, but cannot guarantee a factor better than~$ \frac{\maxdeg-1}{2\maxdeg-3} $.
For~$ \maxdeg=3 $ the factor is exactly~$ \frac{2}{3} $, as is also shown in~\cite{b14}.

\medskip

\inlineheading{Our Contributions.}
What is the benefit of using degree information when picking the next edge?
We show tight approximation guarantees for \karpsipser and \mingreedy on bipartite graphs, assuming worst case instead of randomized tie breaking.

We introduce the class of deterministic \emph{degree sensitive} greedy algorithms and show that \karpsipser, \mingreedy, \greedy, \mrg, \shuffle, and all algorithms for the \emph{query commit problem}~\cite{mr11} belong to this class.
(We also consider a class of `two-sided' algorithms like e.g.\ \mds, which repeatedly picks an edge with \emph{minimum degree sum}.)
Our main result is that \mingreedy and \karpsipser are optimal degree sensitive algorithms.

\begin{theorem}
\label{thm:karpsipser}
The \karpsipser algorithm always computes a matching of size at least~$ \theratio $ times optimal for any bipartite graph with degrees at most~\maxdeg.
\end{theorem}

Observe that the guarantee~\theratio for the \karpsipser algorithm implies at least the same guarantee for \mingreedy.
If a degree-1 nodes exist, then both algorithms proceed identically, otherwise the \karpsipser algorithm picks an arbitrary edge whereas \mingreedy employs a finer edge selection routine.

On general graphs, \karpsipser and \mingreedy do not perform equally well.
For~$ \maxdeg{=}3 $, \mingreedy achieves guarantee~\twoOverThree, see~\cite{b14}, whereas \karpsipser can only guarantee~\oneOverTwo (the chord of a length-four cycle might be picked).

\medskip

It is optimal to pick an edge with a degree-1 node, since such an edge belongs to some maximum matching.
This observation in a sense explains \karpsipser.
To prove \cref{thm:karpsipser} we devise a charging scheme which implicitly builds upon this fact.
Consider the connected components of the graph~\mg on edge set~$ \m\cup\mopt $, where \m is the matching computed by \karpsipser and~\mopt is an arbitrary maximum matching.
Connected components of~\mg with small ``local'' approximation ratios are amortized by ``neighboring'' components with large local approximation ratios, where two components are neighbors if they are connected by an edge of the input graph.
When a node gets matched, a charge depending on its current degree is applied.
A node which gets matched when it has degree one is not charged, and has the potential to increase the local approximation ratio of its own or of a neighboring component.

\smallskip
\medskip

To study limitations of greedy matching algorithms we utilize the framework of \emph{adaptive priority algorithms} introduced by Borodin, Nielsen and Rackoff \cite{bnr03}.
It was successfully applied to e.g. Scheduling \cite{bnr03}, Max-Sat \cite{p11}, Sum-Coloring \cite{biyz12}, graph problems like Steiner-Tree or Independent-Set \cite{di09}, or matching in general graphs \cite{b14,p12}.
Inapproximability results are obtained similar to the adversarial arguments found in the analysis of competitive ratios of online algorithms.

An adaptive priority algorithm~$ A $ is defined relative to the notion of a \emph{data item}, in which only part of the input is revealed.
At the beginning of each round~$ A $ computes, incorporating all information gathered in previous rounds, a total priority order of all possible data items and receives the data item~$ d $ of highest priority contained in the input.
Then~$ A $ has to make an irrevocable decision based on~$ d $, thereby constructing part of the solution once and forever.

The notion of ``greedy'' is captured by the submitted orders and the irrevocable decisions.
Adaptive priority algorithms have no resource constraints, hence inapproximability results apply to correspondingly large classes of algorithms.

\medskip

We define \emph{degree sensitive} algorithms which utilize data items of the form
\[ \dataitem{u,d_u,v}\,, \]
where~$ u,v $ are nodes and~$ d_u{\geq}1 $ is an integer.
In any data item~\dataitem{u,d_u,v} received by algorithm~$ A $ nodes~$ u $ and~$ v $ are neighbors and~$ u $ has degree~$ d_u $.
Here we refer to the \emph{reduced graph}, which contains exactly the edges incident with nodes not matched in earlier rounds.
If~\dataitem{u,d_u,v} is received, then~$ u $ and~$ v $ must be matched.

Additionally, before the first round an algorithm may access a priori knowledge on the input.
We allow access to the number of nodes in the input graph.

\begin{theorem}
\label{thm:degsens}
For each degree sensitive algorithm~$ A $ and for any~$ \varepsilon>0 $, there is a bipartite graph of degree at most~\maxdeg (and with a perfect matching) such that~$ A $ computes a matching of size at most~$ \theratio+\varepsilon $ times optimal.
\end{theorem}

Consequently, \karpsipser is an optimal degree sensitive algorithm.
Why?
To implement \karpsipser as a degree sensitive algorithm, in each round the priority order begins with all possible data items~\dataitem{u,1,v}, in arbitrary order, and continues with all remaining data items, also in arbitrary order.
Similarly, \greedy, \mrg, \mingreedy, and \shuffle can be implemented as degree sensitive algorithms.
All algorithms for the query commit problem are degree sensitive as well:
such an algorithm has access to the set of nodes of a graph but has no knowledge of its edges, repeatedly tests whether two unmatched nodes are connected by an edge, and adds each found edge to the matching.

\medskip

Consider so called two-sided algorithms like e.g.\ \mds, which repeatedly picks an edge such that the degree sum of both incident nodes is minimum.
Such algorithms are contained in the natural generalization of degree sensitive algorithms to data items of the form
\[\dataitem{u,d_u,v,d_v}\,,\]
where~$ u,d_u $, and~$ v $ are defined as before and~$ d_v $ is the current degree of node~$ v $.
We show that such algorithms cannot achieve approximation ratio larger than~$ \frac{\maxdeg+1}{2\maxdeg-2} $.

Note that this bound is only marginally weaker than our~$ \frac{\maxdeg}{2\maxdeg-2} $ bound for degree sensitive algorithms, and we conjecture that it can be strengthened to the same factor.
To support our conjecture, we prove it for~$ \maxdeg=3 $ and show that the approximation ratio of \mds is bounded by~$ \frac{\maxdeg}{2\maxdeg-2} $.

\medskip

\inlineheading{Structure of the Paper.}
We prove \cref{thm:karpsipser,thm:degsens} in \cref{sect:mingreedyguarantee,sect:adaprio_upper_bound}, respectively.
In \cref{amoregeneralclassofalgorithms} we discuss two-sided algorithms.
Results on graphs with bounded average degree are discussed in \cref{app:avgdeg}.
Conclusions and open problems are presented in \cref{sect:conclusion}.


\newcommand{\samesidelemmatext}{
If there is a right degree-1 path endpoint~$ w $, then in the right partition there is a degree-1 node which is not a path endpoint.
}
\newcommand{\independentendsteplemmatext}{
The increases of nodes $x_R', x_L$, and  $w_R$ sum up to at least $2$.
}
\newcommand{\dynamictransferlemmatext}{
Consider the donation transfer~\dedge{u_1,v_0}.
\begin{enumerate}[topsep=0mm,noitemsep,label=\alph*)]
\item\label{lemma:dynamictransfer:maxdeg-3}
If~\dedge{u_1,v_0} moves at most~$ \maxdeg-3 $ coins, then~\dedge{u_1,v_0} adheres to our plan.
\item\label{lemma:dynamictransfer:maxdeg-2}
Assume that when~$ u_1 $ is selected it is the only degree-1 node in its partition which is not an augmenting path endpoint.
If~\dedge{u_1,v_0} moves at most~$ \maxdeg-2 $ coins, then~\dedge{u_1,v_0} adheres to our plan.
\end{enumerate}
}
\newcommand{\numstatictransferslemmatext}{
An augmenting path endpoint~$ w $ receives exactly two common credits if and only if~$ w $ has degree~$ d(w)\geq2 $ after its~\mopt-edge is removed from~\g.
Otherwise~$ w $ receives exactly one common credit.
}

\section{A Tight Performance Guarantee for \karpsipser}%
\label{sect:mingreedyguarantee}%
This section proves \cref{thm:karpsipser}.
Let~$ \g{=}(L\cup R,E) $ be the bipartite input graph.
We fix the matching~$ \m{\subseteq} E $ computed by \karpsipser and a maximum matching~$ \mopt{\subseteq} E $.
Nodes in the graph~$ \mg{=}(L\cup R,\m\cup\mopt) $ have degree at most two:
The connected components of~\mg are paths, cycles, and isolated nodes.
We ignore isolated nodes.
W.l.o.g. we choose~\mopt such that each~\mg-component is either an augmenting \emph{path} or a \emph{singleton}.
A path~$ X $ alternates between~$ \w{X}\geq1 $ edges of~\m and~$ \w{X}{+}1 $ edges of~\mopt.
The two path \emph{endpoints} of~$ X $ are not covered by~\m
(%
\scalebox{.4}{%
\begin{tikzpicture}
\usetikzlibrary{fixedpointarithmetic}
\usetikzlibrary{decorations}
\usetikzlibrary{decorations.markings}
\usetikzlibrary{decorations.pathmorphing}

\tikzstyle{every node} = [circle, fill=white,draw=black,minimum size=17pt,inner sep=0pt];
\tikzstyle{mylabel} = [rectangle,inner sep=0pt,fill=white,draw=none,minimum size=1pt];
\tikzstyle{ghost} = [mylabel];
\tikzstyle{newopt} = [double=white, double distance=2pt,thick];
\tikzstyle{alg}=[thick,
    postaction={
        decorate,
        decoration={markings,
                    mark= at position 0.5 
                          with
                          {
                            \draw (-0.075,0.15) -- (-0.075,-0.15)
                                  (+0.075,0.15) -- (+0.075,-0.15);
                          }
                    }
                }
];
\tikzstyle{altpath} = [decoration={snake}, decorate];
\tikzstyle{flow} = [>=triangle 45];
\tikzstyle{comp} = [lightgray, line width=6pt];
\tikzstyle{compopt} = [opt,double=lightgray];
\tikzstyle{result} = [rectangle, draw, fill=white!20, text width=10em, text centered, rounded corners, minimum height=6em];

\definecolor{verylightgray}{RGB}{200,200,200}

\tikzstyle{opt} = [
draw=none,
postaction={draw=black,decorate,decoration={curveto,raise= 1.5pt}},
postaction={draw=black,decorate,decoration={curveto,raise=-1.5pt}}
];

\tikzstyle{highlight} = [
preaction={draw=lightgray,line width=6pt}
];

\tikzstyle{intendedPick} = [dotted,dash pattern=on 1pt off 3pt, ultra thick];
\tikzstyle{leftside} = [fill=lightgray!66];
\tikzstyle{selected} = [ultra thick];

\node (1) {};
\node[right=of 1] (2) {};
\node[right=of 2] (3) {};
\node[right=of 3] (4) {};
\node[right=1.5 of 4] (5) {};
\node[right=of 5] (6) {};
\node[right=of 6] (7) {};
\node[right=of 7] (8) {};

\node[mylabel]  at ($(4)!0.5!(5)$) {$ \dots $};

\draw (1) edge[newopt] (2);
\draw (2) edge[alg] (3);
\draw (3) edge[newopt] (4);
\draw (5) edge[newopt] (6);
\draw (6) edge[alg] (7);
\draw (7) edge[newopt] (8);

\end{tikzpicture}%
}%
 where~\mopt-edges and~\m-edges are drawn double resp. crossed%
)%
.
A singleton is an edge contained in both~$ \m$ and~$ \mopt $
(%
\scalebox{.4}{%
\begin{tikzpicture}
\usetikzlibrary{fixedpointarithmetic}
\usetikzlibrary{decorations}
\usetikzlibrary{decorations.markings}
\usetikzlibrary{decorations.pathmorphing}

\tikzstyle{every node} = [circle, fill=white,draw=black,minimum size=17pt,inner sep=0pt];
\tikzstyle{mylabel} = [rectangle,inner sep=0pt,fill=white,draw=none,minimum size=1pt];
\tikzstyle{ghost} = [mylabel];
\tikzstyle{newopt} = [double=white, double distance=2pt,thick];
\tikzstyle{alg}=[thick,
    postaction={
        decorate,
        decoration={markings,
                    mark= at position 0.5 
                          with
                          {
                            \draw (-0.075,0.15) -- (-0.075,-0.15)
                                  (+0.075,0.15) -- (+0.075,-0.15);
                          }
                    }
                }
];
\tikzstyle{altpath} = [decoration={snake}, decorate];
\tikzstyle{flow} = [>=triangle 45];
\tikzstyle{comp} = [lightgray, line width=6pt];
\tikzstyle{compopt} = [opt,double=lightgray];
\tikzstyle{result} = [rectangle, draw, fill=white!20, text width=10em, text centered, rounded corners, minimum height=6em];

\definecolor{verylightgray}{RGB}{200,200,200}

\tikzstyle{opt} = [
draw=none,
postaction={draw=black,decorate,decoration={curveto,raise= 1.5pt}},
postaction={draw=black,decorate,decoration={curveto,raise=-1.5pt}}
];

\tikzstyle{highlight} = [
preaction={draw=lightgray,line width=6pt}
];

\tikzstyle{intendedPick} = [dotted,dash pattern=on 1pt off 3pt, ultra thick];
\tikzstyle{leftside} = [fill=lightgray!66];
\tikzstyle{selected} = [ultra thick];

\node (1) {};
\node[right=of 1] (2) {};
\draw (1) edge[opt,alg] (2);

\end{tikzpicture}%
}%
)%
.
Each other component, i.e. each even-length path or cycle, is turned into singletons by replacing its maximum matching edges with its~\m-edges
(%
\scalebox{.4}{%
\begin{tikzpicture}
\usetikzlibrary{fixedpointarithmetic}
\usetikzlibrary{decorations}
\usetikzlibrary{decorations.markings}
\usetikzlibrary{decorations.pathmorphing}

\tikzstyle{every node} = [circle, fill=white,draw=black,minimum size=17pt,inner sep=0pt];
\tikzstyle{mylabel} = [rectangle,inner sep=0pt,fill=white,draw=none,minimum size=1pt];
\tikzstyle{ghost} = [mylabel];
\tikzstyle{newopt} = [double=white, double distance=2pt,thick];
\tikzstyle{alg}=[thick,
    postaction={
        decorate,
        decoration={markings,
                    mark= at position 0.5 
                          with
                          {
                            \draw (-0.075,0.15) -- (-0.075,-0.15)
                                  (+0.075,0.15) -- (+0.075,-0.15);
                          }
                    }
                }
];
\tikzstyle{altpath} = [decoration={snake}, decorate];
\tikzstyle{flow} = [>=triangle 45];
\tikzstyle{comp} = [lightgray, line width=6pt];
\tikzstyle{compopt} = [opt,double=lightgray];
\tikzstyle{result} = [rectangle, draw, fill=white!20, text width=10em, text centered, rounded corners, minimum height=6em];

\definecolor{verylightgray}{RGB}{200,200,200}

\tikzstyle{opt} = [
draw=none,
postaction={draw=black,decorate,decoration={curveto,raise= 1.5pt}},
postaction={draw=black,decorate,decoration={curveto,raise=-1.5pt}}
];

\tikzstyle{highlight} = [
preaction={draw=lightgray,line width=6pt}
];

\tikzstyle{intendedPick} = [dotted,dash pattern=on 1pt off 3pt, ultra thick];
\tikzstyle{leftside} = [fill=lightgray!66];
\tikzstyle{selected} = [ultra thick];

\node (1) {};
\node[right=of 1] (2) {};
\node[right=of 2] (3) {};
\draw (1) edge[newopt] (2) (2) edge[alg] (3);

\node[right=.667 of 3,mylabel] {\scalebox{2}{$ \leadsto $}};

\node[right=2 of 3] (1) {};
\node[right=of 1] (2) {};
\node[right=of 2] (3) {};
\draw (1) edge[] (2) (2) edge[alg,newopt] (3);

\end{tikzpicture}%
}%
)%
.
Since we ignore isolated~\mg-nodes any node is~\m-covered or a path endpoint, which never gets matched.

\smallskip

\inlineheading{Local Approximation Ratios.}
We lower bound \emph{local} approximation ratios of paths and singletons.
A path~$ X $ has local approximation ratio~$ \frac{\w{X}}{\w{X}+1} $, a singleton has local approximation ratio~$ \frac{1}{1}{=}1 $.
Small local approximation ratios of short paths will be amortized by those of long paths and singletons.

We transfer~\emph{coins} between~\mg-components, each coin is worth~\transferred of \mbox{`\m-funds'}.
If component~$ X $ \emph{receives}~$ c_X $ coins and \emph{pays}~$ d_X $ coins, then~$ c_X{-}d_X $ is the \emph{balance} of~$ X $.
The local approximation ratio of~$ X $ becomes
\[ \left(~~\w{X}+\transferred\cdot(c_X-d_X)~~\right)~~/~~\wopt{X}\,, \]
where $ \w{X},\wopt{X} $ are the numbers of~\m-edges respectively~\mopt-edges of~$ X $.
We establish balances of at least
\begin{align}
\label{eqn:balanceboundsingleton}
c_X-d_X&\geq-2(\maxdeg-2)&\mbox{for each singleton~$ X $ and}\\
\label{eqn:balanceboundaug}
c_X-d_X&\geq-\w{X}\cdot2(\maxdeg-2)+2\maxdeg&\mbox{for each path~$ X $.}
\end{align}
The local approximation ratio of a singleton~$ X $ is at least \mbox{$l:=1-\transferred\cdot2(\maxdeg-2) $}, since we have~$ \w{X}{=}\wopt{X}{=}1 $.
Choosing $ \transferred:=\frac{1}{2(2\maxdeg-2)} $ we obtain a lower bound of~$ l=\theratio $.
The local approximation ratio of a path~$ X $ attains the same lower bound, since it is at least
\begin{align*}
\frac{\w{X}-\w{X}{\cdot}\transferred{\cdot}2(\maxdeg{-}2)+\transferred{\cdot}2\maxdeg}{\w{X}+1}
=\frac{\w{X}\cdot l+\transferred{\cdot}2\maxdeg}{\w{X}+1}
=l+\frac{\transferred{\cdot}2(2\maxdeg{-}2)-1}{\w{X}+1}=l\,.
\end{align*}

Since the minimum local approximation ratio over all components in~\mg is $ l{=}\theratio $, \karpsipser achieves (global) approximation ratio at least~\theratio:
\[ \frac{|\m|}{|\mopt|}=\frac{\sum_X\w{X}}{\sum_X\wopt{X}}=\frac{\sum_X\w{X}+\transferred\cdot(c_X-d_X)}{\sum_X\wopt{X}}\geq\frac{\sum_X\theratio\cdot\wopt{X}}{\sum_X\wopt{X}}=\theratio\,. \]

\subsection{Balance Bounds: The Plan}

To establish \cref{thm:karpsipser} it remains to verify the balance bounds \cref{eqn:balanceboundsingleton,eqn:balanceboundaug}.
Here is our plan.
We claim that each~\m-covered node of a path~$ X $ pays at most~$ \maxdeg-2 $ coins.
Hence the balance of~$ X $ is at least~$ c_X-d_X\geq-2\w{X}\cdot(\maxdeg-2) $.
To verify \cref{eqn:balanceboundaug} we prove a balance \emph{increase} for~$ X $ of at least~$ 2\maxdeg $.
Increase for~$ X $ comes from~$ X $-nodes which pay less than~$ \maxdeg-2 $ coins or receive coins.

The first~$ X $-node~$ u_L\in L $ in the left partition is matched in the \emph{creation step} of~$ X $.
The \emph{left end step} of~$ X $ matches the~$ X $-node~$ x_L{\in} L $ in the edge~$ \edge{x_L,w_R}{\in}\mopt $ with the~$ X $-endpoint~$ w_R\in R $.
Note that~$ u_L{=}x_L $ might hold.
The node matched with~$ x_L $ is called~$ x_R' $.
Nodes in the opposite partitions are defined analogously
(double drawn edges belong to~\mopt and crossed edges belong to~\m):
\begin{center}
\scalebox{.8}{
\begin{tikzpicture}
\usetikzlibrary{fixedpointarithmetic}
\usetikzlibrary{decorations}
\usetikzlibrary{decorations.markings}
\usetikzlibrary{decorations.pathmorphing}

\tikzstyle{every node} = [circle, fill=white,draw=black,minimum size=17pt,inner sep=0pt];
\tikzstyle{mylabel} = [rectangle,inner sep=0pt,fill=white,draw=none,minimum size=1pt];
\tikzstyle{ghost} = [mylabel];
\tikzstyle{newopt} = [double=white, double distance=2pt,thick];
\tikzstyle{alg}=[thick,
    postaction={
        decorate,
        decoration={markings,
                    mark= at position 0.5 
                          with
                          {
                            \draw (-0.075,0.15) -- (-0.075,-0.15)
                                  (+0.075,0.15) -- (+0.075,-0.15);
                          }
                    }
                }
];
\tikzstyle{altpath} = [decoration={snake}, decorate];
\tikzstyle{flow} = [>=triangle 45];
\tikzstyle{comp} = [lightgray, line width=6pt];
\tikzstyle{compopt} = [opt,double=lightgray];
\tikzstyle{result} = [rectangle, draw, fill=white!20, text width=10em, text centered, rounded corners, minimum height=6em];

\definecolor{verylightgray}{RGB}{200,200,200}

\tikzstyle{opt} = [
draw=none,
postaction={draw=black,decorate,decoration={curveto,raise= 1.5pt}},
postaction={draw=black,decorate,decoration={curveto,raise=-1.5pt}}
];

\tikzstyle{highlight} = [
preaction={draw=lightgray,line width=6pt}
];

\tikzstyle{intendedPick} = [dotted,dash pattern=on 1pt off 3pt, ultra thick];
\tikzstyle{leftside} = [fill=lightgray!66];
\tikzstyle{selected} = [ultra thick];
\node (wl) {$ w_L $};
\node[right=.4 of wl] (xr) {$ x_R $};
\node[right=.4 of xr] (xl') {$ x_L' $};
\node[right=.4 of xl',mylabel] (xl'') {};
\node[right=.2 of xl'',mylabel] (dots) {\dots};
\node[right=.2 of dots,mylabel] (ur') {};
\node[right=.4 of ur'] (ur) {$ u_R $};
\node[right=.4 of ur] (ul) {$ u_L $};
\node[right=.4 of ul,mylabel] (ul') {};
\node[right=.2 of ul',mylabel] (dots) {\dots};
\node[right=.2 of dots,mylabel] (xr'') {};
\node[right=.4 of xr''] (xr') {$ x_R' $};
\node[right=.4 of xr'] (xl) {$ x_L $};
\node[right=.4 of xl] (wr) {$ w_R $};

\draw
(wl) edge[newopt] (xr)
(xr) edge[alg] (xl')
(xl') edge[newopt] (xl'')
(ur') edge[newopt] (ur)
(ur) edge[alg] (ul)
(ul) edge[newopt] (ul')
(xr'') edge[newopt] (xr')
(xr') edge[alg] (xl)
(xl) edge[newopt] (wr)
;
\end{tikzpicture}}
\end{center}
Our plan is to show that a balance increase for~$ X $ of at least~\maxdeg can be achieved by some of nodes~$ u_L,x_L,x_R',w_R $ and a certain~\g-neighbor~$ v_R $ of~$ u_L $.
The actual selection of increase nodes is determined later.
We say that increase~\maxdeg is achieved \emph{for partition~$ L $ of~$ X $}.
W.l.o.g. in our analysis we discuss partition~$ L $.
A balance increase of~\maxdeg for partition~$ R $ of~$ X $ is obtained from the analogous set of nodes.

\smallskip

\inlineheading{Transfers.}
We move coins over edges in~$ F=E\setminus(\m\cup\mopt) $, where~$ F $-edges connect ``neighboring'' components of~\mg.
An~$ F $-edge which moves coins is called a \emph{transfer}, and moves coins in exactly one direction.
Therefore we denote a transfer as a directed edge~\dedge{u,w} and call it a \emph{debit} from~$ u $ and a \emph{credit} to~$ w $.
We define \emph{common} transfers and \emph{donation} transfers.

\begin{definition}
\label{def:statictransfer}
Let edge~$ \edge{u,w}\in F $ connect an~\m-covered node~$ u $ with a path endpoint~$ w $.
Then~$ \dedge{u,w} $ is a \emph{common transfer} and moves one coin, iff after the step which matches~$ u $ and removes~\edge{u,w} from~\g the degree of~$ w $ is at most one.
\end{definition}

If~$ u_L $ has a common debit~\dedge{u_L,w}, then after creation of~$ X $ node~$ w $ has become a \emph{degree-1} node, i.e. after creation of~$ X $ node~$ w $ has degree exactly one.
Why?
Before~$ u_L,u_R $ are matched, both are incident with an~\m-edge and an~\mopt-edge.
So when~$ X $ is created, all degrees are at least two, since \karpsipser picks an edge with a degree-1 node if possible.
Furthermore, observe that degrees are decreased by at most one in each step since \g is bipartite.
In particular, in the creation step of~$ X $ the degree of~$ w $ is decreased from exactly two to exactly one.

If after creation of~$ X $ there is a (is no) degree-1 path endpoint among the~\g-neighbors of~$ u_L $, then we say that \emph{a (no) right degree-1 endpoint exists after creation of~$ X $}.
In the (no)-case, some of nodes~$ u_L,x_L,x_R',w_R $ achieve a balance increase of at least~\maxdeg for partition~$ L $ of~$ X $ (\cref{lemma:nodeg1creation}).
To discuss the rest of our plan assume the other case, i.e. that a right degree-1 endpoint exists after creation of~$ X $, call it~$ w $.
A certain~\g-neighbor of~$ u_L $ in the right partition of~$ w $, call it~$ v_R $, pushes the balance increase for partition~$ L $ of~$ X $ to at least~\maxdeg (\cref{lemma:deg1creation}).

\smallskip

\inlineheading{\boldmath How to Choose~$ v_R $?}
Recall that the right path endpoint~$ w $ never gets matched.
After creation of~$ X $, node~$ w $ has degree one, thus \karpsipser matches a \mbox{degree-1} node next.
In particular, by \cref{lemma:samesidedeg1} (shown later) the right partition of~$ w $ also contains \mbox{degree-1} nodes~$ v_1,\dots,v_s,\,s\geq1$ which will get matched.

\begin{proposition}
\label{lemma:samesidedeg1}
\samesidelemmatext
\end{proposition}
\newtheorem*{samesidelemma}{\cref{lemma:samesidedeg1}}

\noindent
We choose~$ v_R $ as the first of~$ v_1,\dots,v_s $ which gets matched.
Note that~$ v_R $ is not necessarily matched in the step after creation of~$ X $, since after creation of~$ X $ partition~$ L $ might contain a degree-1 node as well.

No~$ F $-edges are incident with~$ v_R $ when it gets matched with degree one.
So, by \cref{def:statictransfer}, zero common debits leave~$ v_R $.
Thus~$ v_R $ can increase the balance of its component.
If~$ v_R $ belongs to~$ X $, then we will see that some of~$ u_L,x_L,x_R',w_R,v_R $ achieve increase at least~\maxdeg.
If~$ v_R $ belongs to a component~$ Y{\neq} X $ then we donate the increase for~$ Y $ back to~$ X $ using a donation transfer~\dedge{v_R,u_L}.

\begin{definition}
\label{def:dyntrans}
If~$ v_R $ belongs to another component than~$ u_L $, then edge~$ \dedge{v_R,u_L} $ is a \emph{donation transfer}.
Transfer \dedge{v_R,u_L} moves~$ \maxdeg{-}3 $ coins unless the following holds, in which case it moves~$ \maxdeg{-}2 $ coins:
Before~$ v_R $ gets matched the right partition contains exactly~$ \maxdeg{-}2 $ degree-1 nodes besides~$ v_R $ which are all endpoints.
\end{definition}

Our claim that a path node pays at most~$ \maxdeg{-}2 $ coins holds, as we show now.
(Whenever the component for which a node is defined is not clear from context we use superscripts to indicate the component.)
We first argue that~$ v_R^X \neq v_R^Y$ holds for paths~$ X\neq Y $.
Node~$ v_R^X $ has degree one after creation of~$ X $, hence~$ v_R^X $ is matched before \karpsipser picks an edge without a degree-1 node.
In particular, node~$ v_R^X $ is matched before the next path is created, call it~$ Y $.
But~$ v_R^Y $ is matched after~$ Y $ is created, hence we get~$ v_R^X\neq v_R^Y $.
Consequently, at most one donation debit leaves~$ v_R^X $.
Now recall that~$ v_R^X $ has no common debits, since~$ v_R^X $ is matched with degree one.
Our argument applies in particular if~$ v_R^X $ is a path node.
Thus each path node either pays at most~$ \maxdeg{-}2 $ coins in one donation debit, or one coin in each of at most~$ \maxdeg{-}2 $ common debits.

\subsection{Preliminaries}

We have to verify that the increase of a node of a path~$ X $ is counted either for~$ L $ or for~$ R $, but not for both partitions.
We define node sets which increase the balance for partitions~$ L $ resp.~$ R $ of~$ X $, and argue that they do not intersect.
\begin{itemize}[topsep=0mm,noitemsep]
\item
If~$ u_L{=}x_L $ holds, then we obtain increase from nodes in~$ I_L^{=}{=}\lbrace u_L,w_R \rbrace$.
\item
If~$ u_L{\neq} x_L $ holds, then increase comes from nodes in~$ I_L^{\neq}{=}\lbrace u_L,x_L,x_R',w_R\rbrace $.
\end{itemize}
If a right degree-1 endpoint exists after creation, then~$ I_L^=,I_L^{\neq} $ additionally contain~$ v_R $.
Sets~$ I_R^=,I_R^{\neq} $ are defined analogously, depending on~$ u_R{=}x_R $ resp.~$ u_R{\neq}x_R $.

Observe that we have~$ v_R \notin\lbrace x_R,u_R\rbrace $ and~$ v_L \notin\lbrace x_L,u_L\rbrace $ since a donation transfer source node~$ v $ gets matched when it has degree one whereas an~$ x $-node or~$ u $-node gets matched when it is incident with an~\m-edge and an~\mopt-edge.
One of the following holds:
\begin{itemize}[noitemsep,topsep=0mm]
\item
$ u_L{=} x_L{\wedge} u_R{=}x_R $:
In this case observe that~$ I_L^=\cap I_R^==\emptyset $ holds.
\item
$ u_L{=} x_L {\wedge} u_R{\neq}x_R $ (analogous to~$ u_L{\neq} x_L{\wedge} u_R{=}x_R $):
For~$ L $ we obtain increase from nodes in~$ I_L^= $.
From~$ u_R{\neq x_R} $ we get~$ u_L\neq x_L' $, thus~$ I_L^=\cap I_R^{\neq}=\emptyset $ holds.
\item
$ u_L{\neq} x_L{\wedge} u_R{\neq}x_R $:
Here we have~$ u_R{\neq}x_R'{\wedge} u_L{\neq} x_L' $, therefore~$ I_L^{\neq}\cap I_R^{\neq}=\emptyset $ holds.
\end{itemize}

\smallskip

\inlineheading{\boldmath Isolated Nodes in~\mg.}
Recall that our analysis ignores isolated~\mg-nodes.
Why is our guarantee valid?
Isolated~\mg-nodes are never matched by the \karpsipser algorithm.
We assume that each node which is never matched is a path endpoint.
Hence an isolated~\mg-node might receive but does not pay transfers.
Thus it only decreases but does not increase local approximation ratios.

\subsection{Balance Bounds: The Proof}
\label{sect:proofkarpsipser}

Recall that we use a donation transfer~\dedge{v_R,u_L} only if a right \mbox{degree-1} path endpoint~$ w $ exists after creation of a path~$ X $, where~$ w $ is a~\g-neighbor of~$ u_L $.
If~$ w $ belongs to a component other than~$ X $, then a common transfer~\dedge{u_L,w} goes from~$ u_L $ to~$ w $.
If~$ w $ belongs to~$ X $, then we have~$ w{=}w_R $ and~$ u_L{=}x_L $, i.e. path~$ X $ is created in an end step.
In this case~$ w_R $ receives only one common credit:

\newcommand{\propnumstatcredtext}{Node~$ w_R $ receives exactly one common credit iff~$w_R$ has degree at most one after~$ x_L $ gets matched.
Else $w_R$ receives exactly two common credits.}
\begin{proposition}
\label{prop:numstatcred}
\propnumstatcredtext
\end{proposition}
\newtheorem*{propnumstatcred}{\cref{prop:numstatcred}}

Nodes of an end step increase their path's balance by~2.
In particular, increase~2 is achieved no matter if one of the nodes has a donation debit.

\newcommand{\propindependsteptext}{If~$ u_L\neq x_L $ holds, then~$x_R', x_L, w_R$ achieve increase at least~2.}
\begin{proposition}
\label{prop:independstep}
\propindependsteptext
\end{proposition}%
\newtheorem*{propindependstep}{\cref{prop:independstep}}

\cref{prop:numstatcred,prop:independstep} are shown later.
We are ready to verify the balance bound \cref{eqn:balanceboundaug} for a path~$ X $:
increase~\maxdeg is achieved for each of partitions~$ L, R $ of~$ X $.

\begin{lemma}
\label{lemma:nodeg1creation}
Let~$ X $ be a path.
If no right degree-1 endpoint exists after creation of~$ X $, then nodes in~$ I_L^= $ resp.~$ I_L^{\neq} $ increase the balance of~$ X $ by~\maxdeg.
\end{lemma}
\begin{proof}
Recall that no nodes but~$ u_L,u_R $ get isolated at creation.
Since thereafter also no right degree-1 endpoint exists, no common debit leaves~$ u_L $.
Moreover, recall that no donation debit leaves~$u_L$.
Hence~$ u_L $ increases the balance by~$ \maxdeg{-}2 $.

If we have $u_L {=} x_L$, then after creation of~$ X $ node~$w_R$ remains with at least two incident~$ F $-edges.
Both are common credits to~$ w_R $ and further increase the balance of~$ X $ by~2.
So nodes in~$ \{u_L,w_R\}\subseteq I_L^= $ increase the balance of~$ X $ by~\maxdeg.

Otherwise we have~$u_L {\neq} x_L$.
Using \cref{prop:independstep}, we obtain additional increase at least~2 from~$ x_R',x_L,w_R $.
Here we have~$ \lbrace u_L,x_R',x_L,w_R \rbrace{\subseteq} I_L^{\neq}$.
\end{proof}

\begin{lemma}
\label{lemma:deg1creation}
Let~$ X $ be a path.
If a right degree-1 endpoint exists after creation of~$ X $, then nodes in~$ I_L^= $ resp.~$ I_L^{\neq} $ increase the balance of~$ X $ by~\maxdeg.
\end{lemma}
\begin{proof}
Recall that~$ I_L^=,I_L^{\neq} $ also contain~$ v_R $, since a right degree-1 endpoint exists after creation of~$ X $.
We distinguish four cases, which are restated below before their respective analysis.
Assume that~$ u_L{\neq} x_L $ holds.
If~$ v_R $ is a node in~$ X $, then we have~$ v_R{\neq} x_R' $ or~$ v_R{=}x_R' $, which are the first two cases.
In the third case~$ v_R $ is not a node in~$ X $.
If~$ u_L{=}x_L $ holds, then~$ v_R $ is not a node in~$ X $.
Why?
After creation of~$ X $ all~\m- and~\mopt-edges of~$ X $ but those incident with~$ u_L{=}x_L $ and~$ u_R{=}x_R' $ are still in the graph.
So the only \m-covered~$ X $-node which could have degree one now is the~\mopt-neighbor of~$ x_R' $, call it~$ l $.
But~$ v_R\neq l $, since~$ l $ is in the left partition.

{\bf\boldmath $u_L {\ne} x_L$, $v_R$ in $ X$, $v_R {\ne} x_R'$:}
No common or donation transfer leaves~$v_R$, since~$ v_R $ has degree one when it gets matched and belongs to the same path as~$ u_L $.
Thus~$ v_R $ achieves increase~$\maxdeg-2$ for partition~$ L $ of~$ X $.
Since we have~$ v_R\neq x_R' $, the balance increase of~2 for nodes~$x_R', x_L, w_R$ by \cref{prop:independstep} pushes the total increase to at least~$\maxdeg$.
Observe that we have~$ \lbrace v_R,x_R',x_L,w_R\rbrace\subseteq I_L^{\neq}$.

{\bf\boldmath $u_L {\ne} x_L$, $v_R$ in $ X$, $v_R {=} x_R'$:}
Note that~\edge{x_R',u_L} is an~\mopt-edge of~$ X $.
As in the first case, zero debits leave~$v_R$ and~$ v_R $ achieves increase~$\maxdeg-2$.
So we are done if~$w_R$ receives~2 common credits, since then we have~$ \lbrace v_R,w_R\rbrace\subseteq I_L^{\neq}$.
From here on assume that~$w_R$ receives less than two common credits.
By \cref{prop:numstatcred} node~$ w_R $ receives at least one common credit.
A further increase of~1 is obtained if~$u_L$ or~$x_L$ has less than~$\maxdeg-2$ common debits.
Here we have~$ \lbrace v_R,w_R,u_L,x_L\rbrace\subseteq I_L^{\neq} $.

If both~$u_L$ and~$x_L$ have~$\maxdeg-2$ common debits, then we show a contradiction to \cref{lemma:samesidedeg1}:
we argue that, after~$ x_R',x_L $ are matched, there is a right \mbox{degree-1} node and all right degree-1 nodes are endpoints.
After~$ x_R',x_L $ are matched, the destination endpoints of common debits from~$ u_L,x_L $ have degree at most one.
Node~$ w_R $ has degree at most one as well, by \cref{prop:numstatcred}, since we have assumed that~$ w_R $ receives only one common credit.
So the number of endpoints neighboring~$u_L$ (in~\g) is~$\maxdeg-2$, while~$x_L$ has~$\maxdeg-1$ neighbors (in~\g) which are endpoints.
Therefore after~$ x_R',x_L $ are matched an endpoint neighbor of~$ x_L $ (in~\g) has degree one.
Also, all degree-1 nodes in the right partition are endpoints.

{\bf\boldmath $u_L {\ne} x_L$, $v_R$ not in $ X$:}
At most~$\maxdeg{-}3$ common transfers leave~$u_L$, since no common transfer goes from~$u_L$ to~$v_R$.
Therefore~$ u_L $ achieves an increase of~1.
Observe that after creation at most~$\maxdeg-3$ \mbox{degree-1} endpoints exist in the right partition.
Hence by \cref{def:dyntrans}, a donation transfer~\dedge{v_R,u_L} moves~$ \maxdeg{-}3 $ coins to~$ X $.
Using the increase of~2 for nodes~$ x_R',x_L,w_R $ due to \cref{prop:independstep}, the total increase is~\maxdeg.
The increase is obtained from nodes~$ \lbrace u_L,v_R,x_R',x_L,w_R\rbrace\subseteq I_L^{\neq} $.

{\bf\boldmath $u_L {=} x_L$ ($ v_R $ not in $ X $):}
Again, at most~$ \maxdeg-3 $ common debits leave~$ u_L $.
Recall that each destination node of a common debit from~$ u_L $ has degree exactly one after creation of~$ X $.
Also, node~$ w_R $ has degree one after creation if and only if~$ w_R $ receives exactly one common credit, as \cref{prop:numstatcred} shows.

Assume that~$ w_R $ receives two common credits or~$ u_L $ has at most~$ \maxdeg-4 $ common debits, in which case the increases of~$ u_L $ and~$ w_R $ sum up to at least three, since~$ w_R $ receives at least one common credit by \cref{prop:numstatcred}.
After creation the right partition contains at most~$ \maxdeg-3 $ degree-1 endpoints.
By \cref{def:dyntrans}, a donation transfer~\dedge{v_R,u_L} moves additional~$ \maxdeg-3 $ coins to~$ X $.
We are done with an increase of at least~\maxdeg for partition~$ L $ of~$ X $ by nodes~$ \lbrace u_L,w_R,v_R\rbrace\subseteq I_L^{=} $.

Lastly, assume that~$ w_R $ receives one common credit and~$ \maxdeg-3 $ common debits leave~$ u_L $, i.e. the increases of~$ u_L $ and~$ w_R $ sum up to at least two.
Observe that after creation the right partition contains~$ \maxdeg-2 $ many \mbox{degree-1} endpoints and that~$ v_R $ is the only right degree-1 node which is not an endpoint.
Therefore, by \cref{def:dyntrans}, a donation transfer~\dedge{v_R,u_L} moves additional~$ \maxdeg-2 $ coins to~$ X $.
We get an increase of at least~\maxdeg for~$ L $ of~$ X $ by nodes~$ \lbrace u_L,w_R,v_R\rbrace\subseteq I_L^{=}$.
\end{proof}

Next, we prove that the balance of singletons is large enough.

\begin{lemma}
\label{lemma:singleton}
A singleton pays at most~$ 2(\maxdeg-2) $ coins and therefore satisfies \cref{eqn:balanceboundsingleton}.
\end{lemma}
\begin{proof}
Recall that a node has either common or donation debits, but not both, and at most one donation debit leaves each node.
We distinguish three cases for nodes~$ z_L,z_R $ of a singleton:
both have a donation debit, or both have common debits, or w.l.o.g. a donation debit~\dedge{z_L,u_R} leaves~$ z_L $ and~$ z_R $ has common debits.

\textbf{\boldmath A Donation Debit Leaves Each of~$ z_L,z_R $:}
Exactly two donation debits leave the singleton.
By definition, each moves at most~$ \maxdeg-2 $ coins.

\textbf{\boldmath Both~$ z_L,z_R $ Have Common Debits:}
We show that each of~$ z_L,z_R $ has at most~$ \maxdeg{-}2 $ common debits.
Assume that~$ z_L $ has~$ \maxdeg-1 $ common debits.
When~$ z_L,z_R $ are matched, both are incident with an~$ F $-edge and by definition of \karpsipser all nodes have degree at least two.
Thereafter the destination nodes of common debits from~$ z_L $ have degree one, and these endpoints are the only degree-1 nodes in their partition since the only other~\g-neighbor of~$ z_L $ is~$ z_R $.
A contradiction to \cref{lemma:samesidedeg1}.
An analogous argument applies to~$ z_R $.

\textbf{\boldmath A Donation Debit Leaves~$ z_L $ and~$ z_R $ Has Common Debits:}
We are done if~\dedge{z_L,u_R} moves at most~$ \maxdeg{-}3 $ coins,  since at most~$ \maxdeg-1 $ common debits leave~$ z_R $.
If~\dedge{z_L,u_R} moves~$ \maxdeg{-}2 $ coins, then~$ z_R $ has at most~$ \maxdeg{-}2 $ common debits:
assuming that~$ z_R $ has~$ \maxdeg-1 $ common debits, say to nodes~$ w_L^1,\dots,w_L^{\maxdeg-1} $, we show a contradiction.
By definition of~\dedge{z_L,u_R}, before~$ z_L $ gets matched the partition of~$ z_L $ contains~$ \maxdeg-2 $ degree-1 path endpoints and no other degree-1 nodes but~$ z_L $.
But then after~$ z_L $ is matched, at least one of the~$ w_L^i $ has degree one, since the degree of at most~$ \maxdeg-2 $ endpoints was decreased to zero.
Furthermore, since~$ z_L $ is now matched, all degree-1 nodes in the left partition are path endpoints.
This contradicts \cref{lemma:samesidedeg1}.
\end{proof}

To complete the proof of \cref{thm:karpsipser} we have to show \cref{lemma:samesidedeg1,prop:numstatcred,prop:independstep}.
We start with the result that solely depends on the definition of path endpoints and the bipartiteness of~\g.

\begin{samesidelemma}
\samesidelemmatext
\end{samesidelemma}
\begin{proof}
Assume that all degree-1 nodes in the partition of~$ w $ are path endpoints.
Since these are never matched, an edge with a degree-1 node~$ u $ in the other partition is picked next, say~$ u $ gets matched with~$ v $.
Observe that~$ v $ is in the partition of~$ w $ and that all degrees in this partition, but that of~$ v $, are not changed.
So the set of degree-1 nodes in the partition of~$ w $ remains unchanged.
By repeating the argument the degree of~$ w $ is never decreased to zero.
A contradiction.
\end{proof}

Next, we prove the result on the number of common credits to an endpoint.

\begin{propnumstatcred}
\propnumstatcredtext
\end{propnumstatcred}
\begin{proof}
First, recall that no degree-1 node is matched in the creation step of the path of~$ w_R $.
At creation, node~$ w_R $ is not yet isolated and consequently has degree at least two as well.
Since~\g is bipartite, edges incident with~$ w_R $ are removed in pairwise different steps.
Hence there is a step when~$ w_R $ has degree two.

An edge is not a common credit to~$ w_R $ if it is removed before~$ w_R $ has degree two.
Thereafter, each~$ F $-edge removed from~$ w_R $ is a common credit to~$ w_R $.
Hence if~$ w_R $ has degree two when~$ x_L $ is already matched, then both remaining~$ F $-edges are common credits.
If~$ w_R $ has degree two when~$ x_L $ is not yet matched, then~$ w_R $ has only one incident~$ F $-edge and receives one common credit, and after~$ x_L $ is matched~$ w_R $ has degree at most one.
\end{proof}

\begin{propindependstep}
\propindependsteptext
\end{propindependstep}
\begin{proof}
Observe that no donation debit leaves~$ x_L $, since~$ x_L $ has degree at least two when it is matched.
We distinguish if a donation debit leaves $x_R'$ or not. 

{\bf \boldmath No Donation Debit Leaves~$ x_R' $:}
If~$ w_R $ receives two common credits, then we are done.
Otherwise~$ w_R $ receives exactly one common credit, by \cref{prop:numstatcred}.
Therefore it suffices to find an additional increase of one.
If one of~$ x_R',x_L $ has less than~$ \maxdeg-2 $ common debits, then we are done.
So let each of~$ x_R',x_L $ have~$ \maxdeg-2 $ common debits.
Consequently each of~$ x_R',x_L $ is incident with~$ \maxdeg-2 $ many~$ F $-edges just before being matched, i.e. both their degrees---and hence all degrees---are at least two.
After~$ x_R',x_L $ are matched, the destination nodes of common debits from~$ x_L $ have degree exactly one, since their degrees are decreased by exactly one.
Since~$ w_R $ receives one common credit, node~$ w_R $ also has degree one as a consequence of \cref{prop:numstatcred}.
Hence all degree-1 nodes in the right partition are path endpoints.
A contradiction to \cref{lemma:samesidedeg1}.

{\bf \boldmath A Donation Debit~\dedge{x_R',u_L} Leaves~$ x_R' $:}
Recall that~no common debit leaves~$ x_R' $, since~$ x_R' $ is matched when it has degree one.
If~\dedge{x_R',u_L} moves~$ \maxdeg-3 $ coins, then~$ x_R' $ increases the balance by~1.
Using a common credit to $ w_R $, which exists by \cref{prop:numstatcred}, we get a total increase of at least~2.

Now assume that~\dedge{x_R',u_L} moves~$ \maxdeg-2 $ coins.
If~$ w_R $ receives two common credits, or~$ w_R $ receives one common credit and at most~$ \maxdeg-2 $ common debits leave~$ x_L $, then we are done.
So assume that~$ w_R $ receives one common credit and~$ \maxdeg-2 $ common debits leave~$ x_L $, say to nodes~$ w_R^1,\dots,w_R^{\maxdeg-2} $.
We show a contradiction to \cref{lemma:samesidedeg1}.
After~$ x_R',x_L $ are matched, the~$ w_R^i $ have degree at most one by definition, and~$ w_R $ has degree at most one due to \cref{prop:numstatcred}.
We claim that at least one of~$ w_R $ and the~$ w_R^i $ has degree exactly one after~$ x_R',x_L $ are matched.
Why?
Since~\dedge{x_R',u_L} moves~$ \maxdeg-2 $ coins, before~$ x_R',x_L $ are matched the right partition contains exactly~$ \maxdeg-2 $ \mbox{degree-1} endpoints.
Hence thereafter at most~$ \maxdeg-2 $ of~$ w_R $ and the~$ w_R^i $ are isolated, as claimed.
Furthermore, before~$ x_R',x_L $ are matched node~$ x_R' $ is the only degree-1 node in its partition which is not an endpoint, and thereafter~$ x_R' $ is matched.
So after~$ x_R',x_L $ are matched all degree-1 nodes in the right partition are endpoints.
This contradicts \cref{lemma:samesidedeg1}.
\end{proof}


\section{A Performance Bound for Degree Sensitive Algorithms}
\label{sect:adaprio_upper_bound}
In this section we prove \cref{thm:degsens}.
We describe the \emph{adaptive priority game} between algorithm~$ A $ and an adversary~$ B $, who processes the priority orders submitted by~$ A $ in order to construct a hard input instance.
In each round, adversary~$ B $ presents the highest priority data item~\dataitem{u,d_u,v} in the current order which should be in the graph:
Each presented data item must be \emph{consistent} with the previous construction, i.e. giving the final construction as input to~$ A $ must result in the same sequence of submitted priority orders and received data items.

\smallskip

We first prove our~\theratio bound for bipartite graphs with degrees at most~$ \maxdeg{\geq}4 $ and without a perfect matching.
Thereafter we modify~$ B $ such that the construction also works for~$ \maxdeg=3 $, and such that the graph has a perfect matching.

\smallskip

Adversary~$ B $ constructs a graph which contains~$ k $ \emph{traps}~$ T_1,T_2,\dots,T_k $.
For each trap~$ T_i $ algorithm~$ A $ will insert~$\maxdeg$ edges into its matching (crossed edges in \cref{fig:threadSmallLollipops}), whereas~$ T_i $ contains~$2 \maxdeg {-} 2$ edges of a maximum matching (double edges).
Besides traps the graph contains a constant number of additional nodes and edges.
Hence~$ A $ achieves approximation ratio at most~$ \theratio{+}\varepsilon $ for large~$ k $.

Trap~$ T_i $ contains a \emph{left cycle} on nodes~$ c_1^i,c_2^i,c_3^i,c_4^i $ which is connected via an edge~\edge{c_1^i,p_1^i} to a \emph{left path} on nodes~$ p_1^i,p_2^i $.
Trap~$ T_i $ also contains a \emph{right cycle} on nodes $ d_1^i,d_2^i,d_3^i,d_4^i$ connected via~\edge{d_1^i,q_1^i} to a \emph{right path} on nodes~$q_1^i,q_2^i $.
The left path is connected to the right cycle via edges~$ \edge{p_2^i,d_1^i},\edge{p_2^i,d_3^i} $, and analogously the right path of~$ T_i $ is connected to the left cycle of the next trap~$ T_{i+1} $ via edges~$ \edge{q_2^i,c_1^{i+1}},\edge{q_2^i,c_3^{i+1}} $;
the right path of the last trap~$ T_k $ is connected to an extra cycle on nodes~$ e_1,e_2,e_3,e_4 $ via edges $\lbrace q^k_2, e_1 \rbrace, \lbrace q^k_2, e_3 \rbrace $;
an extra node~$ e_0 $ connects to the left cycle nodes~$ c_1^1,c_3^1 $ of the first trap.
The left and right cycles in~$ T_i $ are connected by~$ \Lambda= \maxdeg{-}4 $ many length-three paths on nodes~$ w_j^i,x_j^i,y_j^i,z_j^i $ via edges $ \edge{c_1^i,w_j^i},\edge{c_3^i,w_j^i} $ and $ \edge{z_j^i,d_1^i},\edge{z_j^i,d_3^i} $ for~$ 1\leq j\leq \Lambda $.
During the game~$ B $ will add more edges to this graph, depending on the actions taken by~$ A $.

\begin{figure}[htbp!]
\centering
\large
\scalebox{.55}{
\begin{tikzpicture}[x=.75cm]
\input{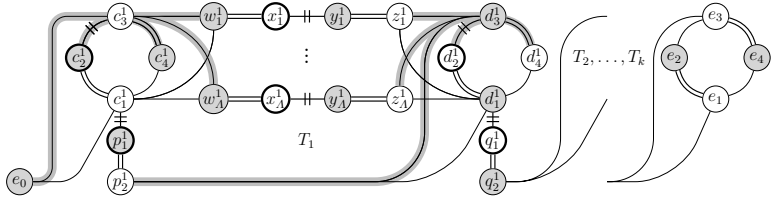}
\tikzstyle{every node} = [circle, fill=white,draw=black,minimum size=18pt,inner sep=.25pt];
\node[leftside] (p2s) {$ e_0 $};

\node[] at ($ (p2s)+(3.25,2) $) (c11) {$c_1^1$};
\node[selected,leftside,above of=c11,left  of=c11] (c21) {$c_2^1$};
\node[above of=c11,above of=c11] (c31) {$c_3^1$};
\node[leftside,above of=c11,right of=c11] (c41) {$c_4^1$};
\node[leftside,selected] (p11) at ($ (c11)+(0,-1) $) {$p_1^1$};
\node[] (p21) at ($ (p11)+(0,-1) $) {$p_2^1$};

\draw
(c11) edge[newopt, bend left] (c21)
(c21) edge[highlight,bend left,alg] (c31)
(c31) edge[highlight,newopt,double=lightgray, bend left] (c41)
(c41) edge[bend left] (c11)
(c11) edge[alg] (p11)
(p11) edge[newopt] (p21)
;

\node[leftside,] (c12) at ($ (c11)+(12,0) $) {$d_1^1$};
\node[selected,above of=c12,left  of=c12] (c22) {$d_2^1$};
\node[leftside,above of=c12,above of=c12] (c32) {$d_3^1$};
\node[above of=c12,right of=c12] (c42) {$d_4^1$};
\node[selected] (p12) at ($ (c12)+(0,-1) $) {$q_1^1$};
\node[leftside,] (p22) at ($ (p12)+(0,-1) $) {$q_2^1$};

\draw
(c12) edge[newopt, bend left] (c22)
(c32) edge[highlight,newopt,double=lightgray, bend left] (c42)
(c42) edge[bend left] (c12)
(c12) edge[alg] (p12)
(p12) edge[newopt] (p22)
;

\coordinate[right=3.25 of c12] (c13) {};
\coordinate[above of=c13,above of=c13] (c33) {};
\coordinate (p13) at ($ (c13)+(0,-1) $) {};
\coordinate (p23) at ($ (p13)+(0,-1) $) {};

\node[] at ($ (c13)+(3.5,0) $) (c1e) {$ e_1 $};
\node[leftside,above of=c1e,left of=c1e] (c2e) {$ e_2 $};
\node[above of=c1e,above of=c1e] (c3e) {$ e_3 $};
\node[leftside,above of=c1e,right of=c1e] (c4e) {$ e_4 $};

\draw
(c1e) edge[newopt,bend left] (c2e)
(c2e) edge[bend left] (c3e)
(c3e) edge[newopt, bend left] (c4e)
(c4e) edge[bend left] (c1e)
;

\draw[highlight,rounded corners=10pt]
(p2s) -- ($ (p2s)+(1.05,0) $) -- ($ (c31)+(-2.2,0) $) -- (c31)
;
\draw[rounded corners=10pt]
(p2s) -- ($ (p21)+(-1.5,0) $) -- (c11)
;
\draw[highlight,rounded corners=34pt]
(p21) -- ($ (p22)+(-2.2,0) $) -- ($ (c32)+(-2.2,0) $) -- (c32)
;
\draw[rounded corners=30pt]
(p21) -- ($ (p22)+(-1.5,0) $) -- (c12)
(p22) -- ($ (p23)+(-1.5,0) $) -- ($ (c33)+(-1.5,0) $) -- (c33)
(p22) -- ($ (p23)+(-1.5,0) $) -- (c13)
(p23) -- ($ (p22)+(5.175,0) $) -- ($ (c3e)+(-2,0) $) -- (c3e)
(p23) -- ($ (p22)+(6,0) $) -- (c1e)
;

\node[mylabel] at ($ (p12)+(3.75,2) $) {$ T_2,\dots,T_k $};

\node[leftside,] (b11) at ($ (c31)+(3,0) $) {$ w_1^1 $};
\node[selected] (b21) at ($ (b11)+(2,0) $) {$ x_1^1 $};
\node[leftside,] (b31) at ($ (b21)+(2,0) $) {$ y_1^1 $};
\node (b41) at ($ (b31)+(2,0) $) {$ z_1^1 $};

\draw
(b11) edge[newopt] (b21)
(b21) edge[alg]
(b31)
(b31) edge[newopt] (b41)
;

\node[leftside,] (b1l) at ($ (b11)+(0,-2) $) {$ w_\Lambda^1 $};
\node[selected] (b2l) at ($ (b1l)+(2,0) $) {$ x_\Lambda^1 $};
\node[leftside,] (b3l) at ($ (b2l)+(2,0) $) {$ y_\Lambda^1 $};
\node (b4l) at ($ (b3l)+(2,0) $) {$ z_\Lambda^1 $};

\draw
(b1l) edge[newopt] (b2l)
(b2l) edge[alg] 
(b3l)
(b3l) edge[newopt] (b4l)
;

\draw
(b11) edge[highlight] (c31)
(b11) edge[out=270,in=0]  (c11)
(b1l) edge[highlight,out=90,in=0] (c31)
(b1l) edge[] (c11)

(b41) edge[highlight] (c32)
(b41) edge[out=270,in=180] (c12)
(b4l) edge[highlight,out=90,in=180]  (c32)
(b4l) edge[] (c12)
;

\draw
(b11) edge[] (c31)
(c22) edge[highlight,bend left,alg] (c32)

(b41) edge[] (c32)

(c31) edge[newopt,double=lightgray, bend left] (c41)
(c21) edge[bend left,alg] (c31)

(b11) edge[out=270,in=0]  (c11)
(b41) edge[out=270,in=180] (c12)
;
\draw[rounded corners=34pt]
(p21) -- ($ (p22)+(-2.2,0) $) -- ($ (c32)+(-2.2,0) $) -- (c32)
;

\node[mylabel] at ($ (b21)+(1,-.85) $) {\huge$ \vdots $};
\node[mylabel] at ($ (b21)+(1,-3) $) {$ T_1 $};
\end{tikzpicture}
}
\caption{The construction of adversary~$ B $.
Algorithm~$ A $ receives data items for bold nodes.
The partitions of the graph are marked with white and gray nodes.
Gray edges form connected components in rounds~$ \Lambda{+}2$ resp.~$\Lambda{+}4 $ of the adaptive priority game.}
\label{fig:threadSmallLollipops}
\end{figure}

To start the game, adversary~$ B $ announces the number~$ k{\cdot}(12+4\Lambda){+}5 $ of nodes.
The construction of~$ B $ proceeds such that after the first~\maxdeg rounds all nodes in~$ T_1 $ but~$ q_2^1 $ are isolated.
The graph to be constructed thereafter is one trap `shorter' with~$ q_2^1 $ instead of~$ e_0 $ connected to the leftmost trap.
Adversary~$ B $ repeats its strategy for~$ T_2,T_3,\dots,T_k $.
After~$ B $ finishes the construction of~$ T_k $, algorithm~$ A $ scores at most two edges for nodes~$ q_2^k,e_1,e_2,e_3,e_4 $.

\smallskip

Observe that in the first round the minimum degree is two.
In each of rounds~$ 1{\leq}j{\leq}\Lambda $, adversary~$ B $ presents the highest priority data item~\dataitem{u,d_u,v} with~$ 2{\leq} d_u{\leq}\maxdeg $ in the respective priority order submitted by~$ A $.
Adversary~$ B $ then relabels nodes in the graph such that~$ u=x_j^1 $ and~$ v=y_j^1 $ holds, i.e. algorithm~$ A $ picks the crossed edges in the length-three paths.

In each round~$ B $ has committed to~$ u=x_j^1 $ having current degree~$ d_u $.
Since~$ d_u $ may be larger than two, adversary~$ B $ inserts additional edges incident with~$ x_j^1 $ into the graph in \cref{fig:threadSmallLollipops}.
The~$ d_u-2 $ additional edges connect~$ x_j^1 $ with arbitrary nodes in the set $\{ w_1^1,\dots,w_{j-1}^1,w_{j+1}^1,\dots,w_{\Lambda}^1,c_2^1,c_4^1,q_2^1 \}$.
This set has cardinality~$ \maxdeg-2\geq d_u-2 $ and only contains nodes outside the partition of~$ x_j^1 $.

The additional edges are consistent:
In previous rounds~$ A $ could not gather knowledge about the neighborhood of~$ u=x_j^1 $, or any other still unmatched node, therefore the additional edges do not have effect on previous actions taken by~$ A $.

Edges incident with~$ u,v $---including additional edges---are removed from the graph in the next round, hence in round~$ j+1 $ the minimum degree is two, again.

The \g-degrees of the nodes receiving additional edges are increased to at most~$ \maxdeg{-}1 $ during rounds~$ 1{\leq}j{\leq}\Lambda $:
The~$ w_j^1 $ have degree at most~$ 3{+}\Lambda{-}1 {=}\maxdeg{-}2$, both $ c_2^1,c_4^1 $ have degree at most~$ 2{+}\Lambda{=}\maxdeg{-}2 $, and~$ q_2^1 $ has degree at most~$ 3{+}\Lambda{=}\maxdeg{-}1$.

\smallskip

In round~$ \Lambda+1 $ adversary~$ B $ again presents the highest priority item~\dataitem{u,d_u,v} with~$ 2\leq d_u\leq\maxdeg $ in the submitted order.
This time~$ B $ relabels nodes such that~$ u{=}p_1^1 $ and~$ v{=}c_1^1 $ (hence~$ A $ picks the crossed edge connecting the left cycle and path), and inserts~$ d_u-2\leq\maxdeg-2$ additional edges connecting~$ u $ with arbitrary nodes in the set~$\{z_1^1,\dots,z_{\Lambda}^1, d_2^1,d_4^1\} $.
The \g-degrees of nodes receiving an additional edge are increased by only one, i.e. they do not exceed~$ 4\leq\maxdeg $.

In round~$ \Lambda{+}2 $ a star centered at~$ c_3^1 $ is disconnected from the rest of the graph.
Since~$ A $ computes a maximal matching, these star nodes get isolates when~$ A $ matches~$ c_3^1 $.
W.l.o.g. we assume that~$ A $ isolates these nodes in round~$ \Lambda{+}2 $.

Similarly, adversary~$ B $ constructs the right cycle and path.
In round~$ \Lambda{+}3 $, algorithm~$ A $ matches~$ u{=}q_1^1 $ with~$ v{=}d_1^1 $, where additional~$ d_u-2 $ edges connect~$ q_1^1 $ with arbitrary nodes in the set~$ \{w_1^2,\dots,w_{\Lambda}^2,c_2^2,c_4^2\} $ of left path and cycle nodes in trap~$ T_2 $.
In round~$ \Lambda{+}4{=}\maxdeg $ a star centered at~$ d_3^1 $ is disconnected from the rest of the graph.
W.l.o.g. again, algorithm~$ A $ scores this edge in this round.

\smallskip

Adversary~$ B $ repeats its strategy for the construction of trap~$ T_2 $.
As before,~\g-degrees of nodes which receive additional edges are not increased above~\maxdeg.
However, we have to pay attention to nodes~$ w_1^2,\dots,w_{\Lambda}^2,c_2^2,c_4^2 $.
For these nodes adversary~$ B $ might already have constructed one additional edge from $q^1_1$.
So additional edges in~$ T_2 $ increase the degrees of these nodes to at most~$ \maxdeg $---and not to at most~$ \maxdeg-1 $ as discussed for~$ T_1 $.
This applies analogously to~$ T_3,T_4,\dots,T_k $.

\smallskip

\textbf{\boldmath $ \maxdeg=3 $:}
Paths on nodes~$ w_j^i,x_j^i,y_j^i,z_j^i $ do not exist.
Left and right paths have four nodes~$ p_1^i,\dots,p_4^i $ resp.~$ q_1^i,\dots,q_4^i $ instead of two, and are still connected to cycles via~$ \edge{p_1^i,c_1^i} $ resp.~$ \edge{q_1^i,d_1^i} $.
Edges~$ \edge{ p_2^i, d_1^i},\edge{ p_2^i, d_3^i} $ and~$ \edge{q_2^i, c_1^{i+1}},\edge{q_2^i, c_3^{i+1}}$ connecting paths with nodes of the `next' cycle are replaced by~$ \edge{ p_4^i, q_2^i},\edge{ p_4^i, d_3^i},$ resp.~$\edge{q_4^i, p_2^{i+1}},\edge{q_4^i, c_3^{i+1}}$.
During the game adversary~$ B $ does not insert any additional edges.
All nodes have degrees two or three.
In particular, nodes~$ p_2^1,p_3^1$ have degree three resp. two:
In the first round~$ B $ presents the highest priority data item~\dataitem{u,d_u,v} and relabels nodes such that~$ A $ picks edge~\edge{p_2^1,p_3^1}, no matter if~$ d_u{=}2 $ or~$ d_u{=}3 $ holds.
The remainder of the left cycle and path in~$ T_1 $ is now separated from the rest of the graph.
Therein~$ A $ can pick at most two edges:
Algorithm $ A $ scores three out of four.
Adversary~$ B $ repeats this construction analogously for edges~\edge{q_2^1,q_3^1}, \edge{p_2^2,p_3^2}, \edge{q_2^2,q_3^2}, $ \dots $,~\edge{q^k_2, q^k_3} and their paths and cycles.
Hence we obtain the claimed convergence to~$ \threeOverFour=\theratio $.

\smallskip

\emph{Note.}
Since in case~$ \maxdeg=3 $ the set of edges is fixed, we can strengthen our bound by giving the algorithm additional a priori knowledge on the input, namely the number of edges.
Moreover, observe that the algorithm cannot counter the adversary's strategy even if the degree sequence of the input graph is given as a~priori knowledge.

\medskip

\textbf{Perfect Matching:}
Adversary~$ B $ replaces the extra node~$ e_0 $ with a length-four cycle~$ C $ and connects~$ c_1^1,c_3^1 $ (resp.~$ p_2^1,c_3^1 $ for~$ \maxdeg=3 $) to different nodes of~$ C $ such that degrees in~$ C $ are two and three and the graph is bipartite (similar to the cycle on~$e_1,\dots,e_4  $).
The construction starts as discussed.
However, when the star centered at node~$ c_3^1 $ is disconnect from the rest of the traps, it is still connected to~$ C $.
W.l.o.g. we assume that~$ A $ isolates all nodes still connected to~$ c_3^1 $ in the next rounds.
(To compensate for the two additional edges scored by~$ A $, adversary~$ B $ increases~$ k $.)
Thereafter the construction proceeds as discussed above.

\subsection{A More General Class of Algorithms}
\label{amoregeneralclassofalgorithms}
The class of degree sensitive algorithms is defined based on data items~\dataitem{u,d_u,v}, which state the degree of \emph{one} node in an edge.
The \emph{minimum degree sum} algorithm and the algorithm which selects a minimum degree node and then a minimum degree neighbor use degree information of \emph{both} nodes, hence they cannot be analyzed with the help of our class.

In this section we discuss a generalization of degree sensitive algorithms to `two-sided' algorithms.
Therefore we extend the definition of a data item to~\dataitem{u,d_u,v,d_v}, i.e.\ we allow an algorithm to specify the degrees~$ d_u $ and~$ d_v $ of both nodes~$ u $ and~$ v $ of an edge.
Otherwise, two-sided algorithms are defined exactly like degree sensitive algorithms.

\medskip

We conjecture that no two-sided algorithm can perform better than \karpsipser.
In this section we support this conjecture by showing three related inapproximability bounds.

\medskip

First, we prove that for~$ \maxdeg=3 $ two-sided algorithms cannot beat the performance of degree sensitive algorithms, i.e.\ they are bounded by the same approximation ratio.

\begin{theorem}
\label{thm:doublydegsens}
Consider a two-sided algorithm~$ A $.
For any~$ \varepsilon>0 $ there is a bipartite graph with degrees at most~$ \maxdeg=3 $ (and a perfect matching) such that~$ A $ computes a matching of size at most~$ \threeOverFour+\varepsilon=\theratio+\varepsilon $ times optimal.
\end{theorem}
\begin{proof}
We slightly change the adversary~$ B $ from the proof of \cref{thm:degsens} to obtain an adversary~$ B' $ for~$ A $.
Adversary~$ B' $ removes right paths and cycles and their incident edges from all traps, and connects the path node~$ p_4^i $ to nodes~$ c_3^{i+1},p_2^{i+1} $ of the next cycle and path.
Cycle~$ C $ and the cycle on nodes~$ e_1,e_2,e_3,e_4 $ along with all their incident edges are replaced by two edges in the first and the~$ k $-th cycle and path, namely edges~$ \edge{c_3^1,p_1^1},\edge{c_2^1,p_2^1} $ resp.~$ \edge{c_4^k,p_4^k},\edge{p_1^k,p_4^k} $, see \cref{fig:maxdeg3lollipops}.

\begin{figure}[htbp!]
\begin{minipage}[b]{.7\textwidth}
\centering
\scalebox{.5}{
\begin{tikzpicture}
\input{tikzStyles}

\node[] (h1) {$ p_4^1 $};
\node[above of=h1] (l11) {$ p_3^1 $};
\node[above of=l11] (l12) {$ p_2^1 $};
\node[above of=l12] (l13) {$ p_1^1 $};
\node[above of=l13] (l14) {$ c_1^1 $};
\node[above of=l14,left  of=l14] (l15) {$ c_2^1 $};
\node[above of=l14,right of=l14] (l16) {$ c_4^1 $};
\node[above of=l14,above of=l14] (l17) {$ c_3^1 $};

\node[ellipse,right of=h1, right of=h1, right of=h1, right of=h1] (h2) {$ p_4^2 $};
\node[above of=h2] (l21) {$ p_3^2 $};
\node[above of=l21] (l22) {$ p_2^2 $};
\node[above of=l22] (l23) {$ p_1^2 $};
\node[above of=l23] (l24) {$ c_1^2 $};
\node[above of=l24,left  of=l24] (l25) {$ c_2^2 $};
\node[above of=l24,right of=l24] (l26) {$ c_4^2 $};
\node[above of=l24,above of=l24] (l27) {$ c_3^2 $};

\coordinate[right of=h2, right of=h2, right of=h2, right of=h2] (h3);
\coordinate[above of=h3] (l31);
\coordinate[above of=l31] (l32);
\coordinate[above of=l32] (l33);
\coordinate[above of=l33] (l34);
\coordinate[above of=l34,left  of=l34] (l35);
\coordinate[above of=l34,right of=l34] (l36);
\coordinate[above of=l34,above of=l34] (l37);

\node[] at ($ (h3)+(6,0) $) (h4) {$ p_4^k $};
\node[above of=h4] (l41) {$ p_3^k $};
\node[above of=l41] (l42) {$ p_2^k $};
\node[above of=l42] (l43) {$ p_1^k $};
\node[above of=l43] (l44) {$ c_1^k $};
\node[above of=l44,left  of=l44] (l45) {$ c_2^k $};
\node[above of=l44,right of=l44] (l46) {$ c_4^k $};
\node[above of=l44,above of=l44] (l47) {$ c_3^k $};

\node[draw=none] (hl-1) at ($ (h3)+(2,0) $) {};

\draw
(l23) edge[highlight] (l24)
(l24) edge[highlight,bend left] (l25)
(l24) edge[highlight,bend right] (l26)
(l26) edge[highlight,bend right] (l27)
(l25) edge[highlight,bend left] (l27)
;

\draw[newopt]
(h1) -- (l11)
(l12) -- (l13)
(l14) edge[bend left,newopt,double=lightgray,highlight] (l15)
(l16) edge[bend right,newopt,double=lightgray,highlight] (l17)

(h2) -- (l21)
(l22) -- (l23)
(l24) edge[bend left,newopt,double=lightgray] (l25)
(l26) edge[bend right,newopt,double=lightgray] (l27)

(h4) -- (l41)
(l42) -- (l43)
(l44) edge[bend left,newopt] (l45)
(l46) edge[bend right,newopt] (l47)
;

\draw[highlight,rounded corners=10pt]
(h1) -- ($ (l27)+(-1,0) $) -- (l27)
;

\draw[rounded corners=10pt]

(h1) -- (l22)
(h1) -- ($ (l27)+(-1,0) $) -- (l27)

(h2) -- (l32)
(h2) -- ($ (l37)+(-1,0) $) -- (l37)

(l21) edge[alg] (l22)
(l23) edge[alg] (l24)
(l24) edge[bend right] (l26)
(l25) edge[bend left, alg] (l27)

(l41) edge[alg] (l42)
(l43) edge[alg] (l44)
(l44) edge[bend right] (l46)
(l45) edge[bend left, alg] (l47)
;
\draw[rounded corners=10pt,highlight]
(l17) -- ($ (l17)+(-2,0) $) -- ($ (l13)+(-2,0) $) -- (l13)
;
\draw[rounded corners=10pt]
(l15) -- ($ (l12)+(-1,0) $) -- (l12)
(l46) -- ($ (h4)+(1,0) $) -- (h4)
(l43) -- ($ (l43)+(.5,-.5) $) -- ($ (h4)+(.5,.5) $) -- (h4)

(l11) edge[alg] (l12)
(l13) edge[alg,highlight] (l14)
(l14) edge[bend right,highlight] (l16)
(l15) edge[alg,highlight,bend left] (l17)
;

\node[mylabel] at ($ (l37)+(+1,-2) $) {$ \dots $};

\coordinate[] (h-1) at ($ (h2)+(6,0) $) ;

\draw[rounded corners=10pt]
(h-1) -- ($ (l47)+(-1,0) $) -- (l47)
(h-1) -- (l42)
;
\end{tikzpicture}
}
\caption{The construction of adversary~$ B' $.
No additional edges are inserted during the game.}
\label{fig:maxdeg3lollipops}
\end{minipage}\hfill
\begin{minipage}[b]{.25\textwidth}
\centering
\scalebox{.5}{
\begin{tikzpicture}
\input{tikzStyles}
\node (c1) {};
\node[above of=c1,left of=c1] (c2) {};
\node[above of=c1,above of=c1] (c3) {};
\node[above of=c1,right of=c1] (c4) {};
\node[below of=c1,below of=c1,below of=c1,below of=c1] (c5) {};
\node[above of=c5,left of=c5] (c6) {};
\node[above of=c5,above of=c5] (c7) {};
\node[above of=c5,right of=c5] (c8) {};

\draw
(c3) edge[highlight,newopt,bend left] (c4)
(c4) edge[highlight,,bend left] (c1)
(c7) edge[highlight,newopt,bend left] (c8)
(c8) edge[highlight,,bend left] (c5)
(c4) edge[highlight,] (c8)
;

\draw
(c1) edge[newopt,bend left] (c2)
(c2) edge[,bend left] (c3)
(c3) edge[newopt,bend left,double=lightgray] (c4)
(c4) edge[,bend left] (c1)
(c5) edge[newopt,bend left] (c6)
(c6) edge[,bend left] (c7)
(c7) edge[newopt,bend left,double=lightgray] (c8)
(c8) edge[,bend left] (c5)
(c2) -- (c6)
(c4) -- (c8)
;

\node[mylabel] (spacer) at ($ (c5)+(0,-1) $) {};
\end{tikzpicture}
}
\caption{An additional component.}
\label{fig:extracomp}
\end{minipage}
\end{figure}

Before the first round, adversary~$ B' $ announces that the number of nodes is~$ 8n $ for some large integer~$ n $.
The parameter~$ k $ will be determined by~$ B' $ based on the actions taken by~$ A $.
In particular, the graph has~$ n-k $ additional connected components, each with two length-four cycles connected by two edges like in \cref{fig:extracomp}.
Observe that any edge in any connected component is incident either with a degree-2 node and a degree-3 node or with two degree-3 nodes.
No edge is incident with two degree-2 nodes.

The following \invariant holds throughout the game:
At the beginning of round~$3i+1$ there is an integer~$ 0\leq k^*\leq i $ such that algorithm~$A$ has matched or isolated all nodes in~$ i-k^* $ additional components as well as all nodes but~$ p_4^{k^*} $ in the first~$k^*$ paths and cycles.
No other nodes are matched or isolated.
For~$ i=0 $, before the first round no nodes are isolated and the \invariant holds.

Consider round~$3i+1$.
Observe that the minimum degree is two, and that every edge is incident with at least one node of degree three.
Adversary~$ B' $ presents the highest priority data item~\dataitem{u,d_u,v,d_v} in the order submitted by~$ A $ with~$ d_u,d_v\in\{2,3\} $ and at least one of~$ d_u,d_v $ equals three.

If~$ d_u{=}d_v{=}3 $, then~$ B' $ constructs the next additional component and relabels nodes such that~$ u,v $ are the two leftmost nodes in \cref{fig:extracomp}.
Observe that~$ A $ scores at most three out of four edges in this component, since therein only gray edges are left.
W.l.o.g. we assume that~$ A $ scores the additional two edges in the next two rounds.
Since~$ k^* $ is not increased, the \invariant continues to hold.

If~$ d_u\neq d_v $, then w.l.o.g. let~$ d_u=3 $ and~$ d_v=2 $.
By the \invariant, algorithm~$ A $ has already matched nodes~$ p_2^1,p_3^1,p_2^2,p_3^2,\dots,p_2^{k^*},p_3^{k^*} $ in previous rounds.
Adversary~$ B' $ relabels nodes such that~$ u{=}p_2^{k^*{+}1} $ and~$ v{=}p_3^{k^*{+}1} $ hold.
After nodes~$ p_2^{k^*{+}1},p_3^{k^*{+}1} $ are matched, the remainder of the~$k^*{+}1$-th path and cycle is disconnected from the rest of the graph, see gray edges in \cref{fig:maxdeg3lollipops}.
In this remainder~$ A $ scores at most two more edges.
W.l.o.g. we assume that~$ A $ does so in the next two rounds.
Hence~$ k^* $ is incremented by one and the \invariant holds before round~$ 3(i+1)+1 $.

We assume that the last path and cycle resp. the last additional component is solved optimally, i.e. algorithm~$ A $ scores four out of four edges.
In each other path and cycle and in each other additional component, algorithm~$ A $ scores three out of four edges.
Hence~$ B' $ can choose sufficiently large~$ n $ such that the approximation ratio of~$ A $ is at most~$ \frac{(n-1)\cdot3+4}{4}\leq\frac{3}{4}{+}\varepsilon$.
\end{proof}

Next, we show that two-sided algorithms can perform at most marginally better than degree sensitive algorithms, i.e.\ they cannot beat the inapproximability bound~$ \frac{\maxdeg}{2\maxdeg-2} $ considerably.

\begin{theorem}
	\label{thm:adapriodoublydegsensplus}
	Let~$ A $ be a two-sided algorithm.
	There is a bipartite input graph of degree at most~$ \maxdeg\geq3 $ for which~$ A $ computes a matching of size at most~$ \frac{\maxdeg+1}{2\maxdeg-2} $ times optimal.
\end{theorem}
\begin{proof}
	The adaptive priority game between~$ A $ and an adversary~$ B $ lasts for~$ \maxdeg+1 $ rounds.
	Let~$ \delta=\maxdeg-3 $.
	The final construction~\g contains the graph~$ \g' $ depicted in \cref{fig:2pDdgraph} as a subgraph.
	In particular, graph~\g contains additional edges which are not depicted in~$ \g' $, but~\g does not have any additional nodes.
	
	\begin{figure}[htbp!]
		\centering
		\scalebox{.8}{
		\begin{tikzpicture}
		\input{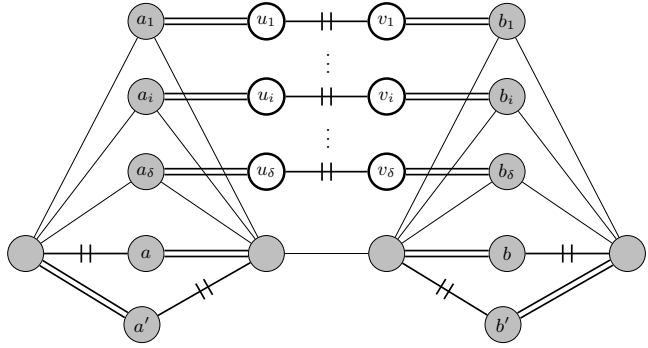}
		\node[fill=lightgray] (c1) {};
		\node[fill=lightgray,right of=c1,right of=c1] (a) {$ a $};
		\node[fill=lightgray,below right=.75 and 1.5 of c1] (a') {$ a' $};
		\node[fill=lightgray,right of=c1,right of=c1,right of=c1,right of=c1] (c2) {};
		\node[fill=lightgray,right of=c2,right of=c2] (c3) {};
		\node[fill=lightgray,right of=c3,right of=c3] (b) {$ b $};
		\node[fill=lightgray,below right=.75 and 1.5 of c3] (b') {$ b' $};
		\node[fill=lightgray,right of=c3,right of=c3,right of=c3,right of=c3] (c4) {};
		
		\node[fill=lightgray,above=3.25 of a] (a1) {$ a_1 $};
		\node[right of=a1,right of=a1,very thick] (u1) {$ u_1 $};
		\node[right of=u1,right of=u1,very thick] (v1) {$ v_1 $};
		\node[fill=lightgray,right of=v1,right of=v1] (b1) {$ b_1 $};
		
		\node[fill=lightgray,above=2 of a] (ai) {$ a_i $};
		\node[right of=a,right of=ai,very thick] (ui) {$ u_i $};
		\node[right of=ui,right of=ui,very thick] (vi) {$ v_i $};
		\node[fill=lightgray,right of=vi,right of=vi] (bi) {$ b_i $};
		
		\node[fill=lightgray,above=.75 of a] (ad) {$ a_{\delta} $};
		\node[right of=ad,right of=ad,very thick] (ud) {$ u_{\delta} $};
		\node[right of=ud,right of=ud,very thick] (vd) {$ v_{\delta} $};
		\node[fill=lightgray,right of=vd,right of=vd] (bd) {$ b_{\delta} $};
		
		\node[mylabel,below right=.1 and .75 of u1] {$ \vdots $};
		\node[mylabel,below right=.1 and .75 of ui] {$ \vdots $};
		
		\draw[newopt]
		(a1) -- (u1)
		(v1) -- (b1)
		(ai) -- (ui)
		(vi) -- (bi)
		(ad) -- (ud)
		(vd) -- (bd)
		(a) -- (c2)
		(c3) -- (b)
		(c1) -- (a')
		(b') -- (c4)
		;
		
		\draw[alg] (u1) -- (v1);
		\draw[alg] (ui) -- (vi);
		\draw[alg] (ud) -- (vd);
		\draw[alg] (c1) -- (a);
		\draw[alg] (b) -- (c4);
		\draw[alg] (a') -- (c2);
		\draw[alg] (c3) -- (b');
		
		\draw
		(a1) -- (c1)
		(a1) -- (c2)
		(ai) -- (c1)
		(ai) -- (c2)
		(ad) -- (c1)
		(ad) -- (c2)
		(b1) -- (c3)
		(b1) -- (c4)
		(bi) -- (c3)
		(bi) -- (c4)
		(bd) -- (c3)
		(bd) -- (c4)
		(c2) -- (c3)
		;
		\end{tikzpicture}
		}
		\caption[Two-sided algorithms: A Hard Instance for $ \maxdeg\geq4 $]{The subgraph~$ \g' $ of the final construction
			(the algorithm receives data items for bold nodes)}
		\label{fig:2pDdgraph}
	\end{figure}
	
	The construction of~$ B $ proceeds such that in rounds~$ 1,\dots,\delta $ algorithm~$ A $ picks edges~$ \edge{u_1,v_1},\dots,\edge{u_{\delta},v_{\delta}} $ and after round~$ \delta $ the reduced graph consists only of gray nodes and of edges connecting gray nodes.
	Observe that all remaining edges touch exactly four gray nodes, namely the unlabeled ones in the figure.
	We assume that in this reduced graph algorithm~$ A $ scores four edges in four rounds, which is optimal.
	Since~$ \g' $ contains a perfect matching of size~$ 2\maxdeg-2 $ and~$ A $ scores one edge in each of~$ \delta+4=\maxdeg+1 $ rounds, the approximation ratio is~$ \frac{\maxdeg+1}{2\maxdeg-2} $, as claimed.
	In the rest of the proof it remains to discuss the first~$ \delta $ rounds.

	\medskip
	
	Recall that~$ A $ does not receive identifiers of neighbors of the nodes in a data item.
	As a consequence, in each round adversary~$ B $ is free---without being inconsistent---to relabel nodes in~$ \g' $ according to the data item presented to~$ A $.
	%
	
	We proceed inductively.
	Assume that at the beginning of round~$ i $ algorithm~$ A $ has picked edges~$ \edge{u_1,v_1},\dots,\edge{u_{i-1},v_{{i-1}}} $.
	The minimum degree in the reduced graph~$ \g_{i} $ is~$ 2 $.
	Adversary~$ B $ uses the set~$ D=\{2,\dots,\maxdeg\} $ of \emph{allowed} degrees:
	From the order submitted by~$ A $ in round~$ i $ adversary~$ B $ presents the highest priority data item~\dataitem{u,d_u,v,d_v} with~$ d_u\in D $ and~$ d_v\in D $.
	Adversary~$ B $ relabels nodes such that~$ u=u_i $ and~$ v=v_i $ hold:
	algorithm~$ A $ picks edge~\edge{u_i,v_i}, as desired.
	
	Now~$ B $ delivers on its promise that both nodes have degree~$ d_u $ resp.\ $ d_v $.
	Therefore~$ B $ inserts additional edges into the graph.
	In particular, since~$ u_i $ already has two incident edges in~$ \g' $, adversary~$ B $ adds~$ d_u-2 $ edges, each incident with~$ u_i $ and one of nodes~$ a,a',a_1,\dots,a_\delta $.
	Analogously, adversary~$ B $ adds~$ d_v-2 $ edges, each incident with~$ v_i $ and one of nodes~$ b,b',b_1,\dots,b_\delta $.
	
	\medskip
	
	It remains to show that~$ B $ does not violate degree constraints when inserting new edges.
	Since we have~$ d_u\leq\maxdeg $ and thus~$ d_u-2\leq\maxdeg-2 $, for all~$ u $-nodes at most~$ \delta(\maxdeg{-}2)=(\maxdeg-3)(\maxdeg{-}2)=(\maxdeg{-}3)^2+(\maxdeg{-}3) $ edges are inserted.
	Since all~$ a $-nodes can receive up to~$ \delta(\maxdeg{-}3)+2(\maxdeg{-}2)=(\maxdeg{-}3)^2+2(\maxdeg{-}2) $ edges, their degrees are increased to at most~\maxdeg if new edges are distributed evenly.
	Analogously, degrees of~$ b $-nodes are at most~\maxdeg.
\end{proof}

Finally, we show that the two-sided \mds algorithm does not achieve better approximation ratio than any degree sensitive algorithm, for all~$ \maxdeg $.

\begin{theorem}
	\label{thm:adapriomds}
	For each~$ \maxdeg\geq3 $ there is a bipartite graph of degree at most~$ \maxdeg $ for which \mds computes a matching of size at most~$ \frac{\maxdeg}{2\maxdeg-2} $ times optimal.
\end{theorem}
\begin{proof}
	\begin{figure}[htbp!]
		\centering
		\scalebox{.75}{
			\begin{tikzpicture}
			\input{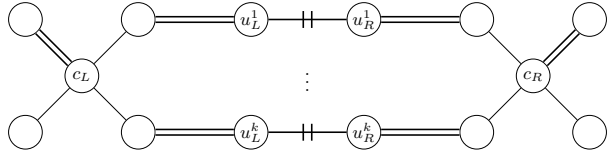}
			
			\node (cl) {$ c_L $};
			\node[left of=cl,above of=cl] (cl1) {};
			\node[left of=cl,below of=cl] (cl2) {};
			
			\node[right of=cl,above of=cl] (w1l) {};
			\node[right of=w1l,right of=w1l] (u1l) {$ u_L^1 $};
			\node[right of=u1l,right of=u1l] (u1r) {$ u_R^1 $};
			\node[right of=u1r,right of=u1r] (w1r) {};
			
			\node[right of=cl,below of=cl] (wnl) {};
			\node[right of=wnl,right of=wnl] (unl) {$ u_L^k $};
			\node[right of=unl,right of=unl] (unr) {$ u_R^k $};
			\node[right of=unr,right of=unr] (wnr) {};
			
			\node[right of=w1r,below of=w1r] (cr) {$ c_R $};
			\node[right of=cr,above of=cr] (cr1) {};
			\node[right of=cr,below of=cr] (cr2) {};
			
			\draw[newopt] (cl1) -- (cl);
			\draw[newopt] (cr1) -- (cr);
			\draw[newopt] (w1l) -- (u1l);
			\draw[newopt] (w1r) -- (u1r);
			\draw[newopt] (wnl) -- (unl);
			\draw[newopt] (wnr) -- (unr);
			
			\draw (cl2) -- (cl);
			\draw (cr2) -- (cr);
			\draw (cl) -- (w1l);
			\draw (cl) -- (wnl);
			\draw (cr) -- (w1r);
			\draw (cr) -- (wnr);
			
			\draw[alg] (u1l) -- (u1r);
			\draw[alg] (unl) -- (unr);
			
			\node[draw=none] at ($ (u1l)+(1,-1) $) {$ \vdots $};
			\end{tikzpicture}
		}
		\caption[\mds: A Hard Instance]{A hard instance for \mds}
		\label{fig:mds}
	\end{figure}
	Choose~$ k=\maxdeg-2 $.
	The hard instance is depicted in \cref{fig:mds}.
	Observe that the degree sum of any edge is at least~4, since nodes~$ c_L $ and~$ c_R $ have degree at least~3.
	Hence we may assume that in the first step \mds picks edge~\edge{u_L^1,u_R^1} (the top crossed edge).
	Assume that edges~$ \edge{u_L^1,u_R^1},\dots,\edge{u_L^i,u_R^i}  $ with~$ i<k $ have already been picked.
	The minimum degree sum is still four, and edge~\edge{u_L^{i+1},u_R^{i+1}} is picked next.
	In the end, for each of nodes~$ c_L $ and~$ c_R $ an incident edge is picked.
	
	Hence the computed matching has~$ k+2=\maxdeg $ edges, whereas a maximum matching consists of the~$ 2k+2=\maxdeg-2 $ double drawn edges.
\end{proof}


\section{Bounded Average Degree}
\label{app:avgdeg}

In \cref{sect:adaprio_upper_bound} we have shown that \mingreedy and \karpsipser achieve the optimal approximation guarantee~\theratio on bipartite graphs with degrees bounded by~\maxdeg.
Our inapproximability results carry over to graphs of bounded average degree.
However, both \mingreedy and the \karpsipser algorithm achieve approximation guarantee only~$ \oneOverTwo+\varepsilon $ even if the average degree is constant.
\begin{theorem}
\label{thm:avgdeg}
The approximation guarantee of \mingreedy and the \karpsipser algorithm is bounded by at most~$ \oneOverTwo+\varepsilon $ for bipartite graphs with average degree at most~$ \frac{7}{2} $, for any~$ \varepsilon>0 $.
\end{theorem}
\noindent
We note that our construction also applies to \greedy, \mrg, \shuffle, the \emph{minimum degree sum} algorithm, the algorithm which first selects a minimum degree node and then a minimum degree neighbor, and to all algorithms for the query commit problem.
\begin{proof}
Nodes of the graph are partitioned into sets~$ L,U,V,W,X,R $, where we have~$ |L|{=}|R|{=}2 $ and~$ Y{=}\{v_Y^1,\dots,v_Y^n\} $ for~$ Y{\in}\{U,V,W,X\} $ and $ n {\in }\mathbb{N} $.
For~$ A,B\in\{U,V,W,X\} $ we denote the set~$ \{\edge{v_A^i,v_B^i}:1\leq i\leq n\} $ as~$ A*B $.
The edge set is
\begin{align*}
(L\times U)~\cup~ (U*V)~\cup~ (V*W)~\cup~ (W*X)~\cup~(X\times R)\,.
\end{align*}
Nodes in~$ L,R $ have degree~$ n $, nodes in~$ U,X $ have degree~$ 3 $, and nodes in~$ V,W $ have degree two.
We argue that any of the given algorithms proceeds as follows, considering worst case tie breaking:
in each of the first~$n$ rounds an edge in~$V*W$ is picked.
Why?
Assuming that only edges in~$ V*W $ have already been picked, the minimum degree over all non-isolated nodes is two;
furthermore, both nodes of each remaining edge in~$ V*W $ have minimum degree degree two.

After round~$ n $ all remaining edges are incident with nodes in~$ L,R $ and the algorithm scores at most four more edges, i.e. a matching of size~$ n+4 $ is computed.
However, observe that $ (U*V) \cup (W*X) $ is a matching of size~$ 2n $.
Therefore an algorithm computes a matching of size at most~$ \frac{n+4}{2n} $ times optimal, which converges to~\oneOverTwo as~$ n\to\infty $.
The average degree in the graph is~$ \frac{2\cdot2\cdot n+2\cdot n\cdot3+2\cdot n\cdot2}{4n+4}=\frac{14n}{4n+4} $, which converges to~$ \frac{7}{2} $ from below.
\end{proof}

\section{Conclusion and Open Problems}
\label{sect:conclusion}
\mingreedy and \karpsipser achieve optimal approximation guarantee~\theratio among degree sensitive algorithms, on bipartite graphs with degrees at most~\maxdeg.

If degree sensitive algorithms are allowed to use data items with degrees of \emph{both} neighbors (`two-sided' algorithms), then we conjecture that the same inapproximability factor~\theratio applies.
However, we can only provide partial proofs, namely for~$ \maxdeg{=}3 $ and for the \mds algorithm.

The \karpsipser algorithm is a refinement of \greedy, since it picks a random edge unless there is a degree-1 node.
What is the expected approximation ratio of the \karpsipser algorithm and the analogous refinement of \mrg?

\bibliography{bib}{}

\end{document}